\documentclass[12pt,a4paper]{article} 
\usepackage{jheppub}

\usepackage{amsthm}
\usepackage{cancel}
\usepackage{empheq}
\usepackage{float}
\usepackage{blindtext}
\usepackage{enumitem}
\usepackage{amsmath,latexsym,amssymb,amsfonts}
\usepackage{bm}
\usepackage{bbm}
\usepackage{mathtools,slashed}
\usepackage{hyperref}
\hypersetup{colorlinks=true,linktoc=page,linkcolor=blue,citecolor=red,urlcolor=cyan}

\usepackage{graphicx}
\usepackage[labelfont=bf]{caption}
\usepackage{xcolor}
\usepackage{appendix}
\numberwithin{equation}{section}
\definecolor{magenta}{RGB}{239, 16, 149}
\hypersetup{
 colorlinks,
 linkcolor={blue},
 citecolor={blue},
 urlcolor={blue}
}
\usepackage{bm}
\usepackage{bbold}
\usepackage{mathtools}
\usepackage{braket}

\numberwithin{equation}{section}

\usepackage{nicefrac}

\usepackage{tcolorbox}

\def \Q{{\mathcal{Q}}}

\newcommand{\tr}{\operatorname{tr}}
\newcommand{\mathi}{{\rm i}}

\title{Resummed effective actions and heat kernels: the Worldline approach and Yukawa assisted pair creation}

\author[a,b]{Filippo~Fecit,}
\author[c,d]{Sebastián~A.~Franchino-Vi\~nas}
\author[e]{and Francisco~D.~Mazzitelli}
\affiliation[a]{\em Dipartimento di Fisica e Astronomia ``Augusto Righi", Universit{\`a} di Bologna,\\
via Irnerio 46, I-40126 Bologna, Italy}

\affiliation[b]{\em INFN, Sezione di Bologna,\\
via Irnerio 46, I-40126 Bologna, Italy}

\affiliation[c]{\em DIME, Universit{\`a} di Genova,\\
via all’Opera Pia 15, 16145 Genova, Italy}

\affiliation[d]{\em INFN, Sezione di Genova,\\
via Dodecaneso 33, 16146 Genova, Italy}

\affiliation[e]{\em Centro At\'omico Bariloche and Instituto Balseiro, Comisi\'on Nacional de Energ\'\i a At\'omica, CONICET,\\
R8402AGP Bariloche, Argentina}

\emailAdd{filippo.fecit2@unibo.it}
\emailAdd{s.franchino-vinas@hzdr.de}
\emailAdd{fdmazzi@cab.cnea.gov.ar}

\abstract{We adapt the Worldline Formalism to obtain resummed expressions for the effective action and the heat kernel of a quantum scalar field coupled to a Yukawa background. The resummation includes all the invariants built from powers, first derivatives and second derivatives of the latter. 
Using such results, we compute the instability of the vacuum computing the vacuum persistence amplitude and the corresponding Schwinger pair production probability. We show that the inclusion of an additional, rapidly varying background can greatly enhance the production of pairs.
}

\keywords{Resummations, heat kernel, effective action, assisted pair production, Worldline Formalism, Schwinger effect}

\begin{document}

\maketitle 





\section*{Conventions}
We work in $D$-dimensional Euclidean spacetime, except for Secs. \ref{sec:non-assisted} and \ref{sec:assisted}, where we Wick rotate to Minkowski spacetime (using the mostly plus signature); we do not distinguish the indices used in these different cases. We use Planck units, so that $\hbar=c=1$.

\section{Introduction} \label{sec:intro}
One could say that the Schwinger effect is as old as it is fascinating. Building on earlier work by Euler and Heisenberg~\cite{Euler:1935zz,Heisenberg:1936nmg}, Schwinger showed in the dawn of the 1950s that the presence of a strong electric field generates an instability in the vacuum state of a fermionic quantum field theory, leading to the nonperturbative creation of particle-antiparticle pairs~\cite{Schwinger:1951nm}.

Despite all the technical advances that have taken place over the last 80 years, an experimental proof of the Schwinger effect has been elusive. The main obstacle to observing this phenomenon is that the critical field required has remained beyond present capabilities. However, the latest generation of laser facilities is expected to overcome this issue in the near future, presently offering field strengths that are just a few orders of magnitude below the critical Schwinger field. In effect, several strong-field experiments are planned or are underway at the European XFEL~\cite{Ahmadiniaz:2024xob, LUXE:2023crk}, LASERIX~\cite{Kraych:2024wwd} and OVAL~\cite{Fan:2017fnd}, to cite a few examples. Additionally, some lower dimensional analogue systems are available to test related effects~\cite{Schmitt:2022pkd}.

The development of appropriate analytical tools to help understand these scenarios is thus indispensable. Recent advances in nonperturbative physics have involved the resurgence theory~\cite{Dunne:2022esi}, large-N expansions~\cite{Karbstein:2023yee, Karbstein:2021gdi}, amplitude techniques~\cite{Copinger:2024pai} and newly devised resummed heat kernel techniques~\cite{Franchino-Vinas:2023wea} (see also Ref.~\cite{Navarro-Salas:2020oew}). 
 
In this article, we will develop an alternative approach, starting from the so-called Worldline Formalism~\cite{Schubert:2001he} to study the interaction of a scalar field with a Yukawa background. 
Relevant to our discussion, the Worldline Formalism has recently been used to perturbatively study combined Yukawa and axial interactions~\cite{Bastianelli:2024vkp}. Moreover, being a functional approach, it naturally offers the possibility to perform nonperturbative analyses. This has been particularly carried out in terms of Worldline instantons~\cite{Dunne:2006st}, which has recently enabled to analyze the spectrum of Schwinger pair creation in spacetime dependent fields~\cite{DegliEsposti:2024upq}. Additionally, the scattering of matter with electromagnetic PP waves has been considered in Ref.~\cite{Copinger:2024twl}. Of course, nonperturbative numerical techniques are also available and can readily be applied in quantum mechanics~\cite{Ahumada:2023iac} or in interacting quantum field theories~\cite{Franchino-Vinas:2019udt}.

More specifically, in Sec.~\ref{sec1} we will pose our problem in terms of the Worldline Formalism.
Afterwards, in Sec.~\ref{sec:resummation} we will show how to perform a resummation, at the level of the heat kernel, of the invariants involving powers, first and second derivatives of the Yukawa background; this is very much in the spirit of the Parker and Toms large Ricci scalar resummation~\cite{Parker:1984dj,Hu:1984js,Jack:1985mw,Flachi:2015sva}. In Sec.~\ref{sec:assisted_yukawa} we introduce an additional fastly varying background, which will exert the rôle of an assisting (or catalyzing) field~\cite{Schutzhold:2008pz}; for this new field, we will use thus the Yukawa analogue of the Barvinsky-Vilkovisky resummation in curved spacetimes~\cite{Barvinsky:1987uw, Barvinsky:1990up, Silva:2023lts}. We will state our conclusions in Sec.~\ref{sec:conclusion}. Finally, details about the boundary conditions and Green functions are given in App.~\ref{appA} and App.~\ref{appB} refers to the computation of functional determinants and the generalized Gel'fand--Yaglom theorem; App.~\ref{app:coeff_SI} contains the generalized Schwinger--DeWitt coefficients computed using the string inspired boundary conditions.

\section{Worldline representation of the scalar heat kernel} \label{sec1}
Consider a theory consisting on a single quantum scalar field $\phi(x)$, in flat spacetime and interacting with a Yukawa background (classical) field, which is described by the action
\begin{equation}
S=\frac12 \int {\rm d}^Dx \left[ (\partial \phi)^2 + V(x) \phi^2\right]\ ,
\end{equation}
where Einstein’s sum convention $\partial^2=\partial_\mu \partial^\mu$ is employed. 
Generalizations including gauge background fields or curved spacetimes are in principle possible. However, in this work we will consider only the scalar coupling through $V(x)$, which will display several important features of our techniques; the generalization to other couplings, as well as higher spins, is left to future publications. 

We will thus focus on $V(x)$, which is an arbitrary Yukawa type (scalar) potential and may naturally include a mass term for the scalar field. 
As customarily, the one-loop effective action $\Gamma$, which is the full effective action unless we quantize $V$ as well, is related to the operator of quantum fluctuations, $\Q:=-\partial^2 + V$:
\begin{equation}
\Gamma=\frac12 \,\mathrm{Log} \mathrm{Det}[\Q]\ .
\end{equation}
In the Schwinger proper-time parametrization~\cite{DeWitt:2003pm} (or Frullani's 
representation for the logarithm~\cite{Jeffreys}), we can write the effective action as
\begin{equation} \label{Gamma}
\Gamma=-\frac12 \int_{0}^{\infty}\frac{{\rm d}T}{T}\int {\rm d}^Dx \, K(x,x;T)\ ,
\end{equation}
in terms of the diagonal of the related heat kernel operator
\begin{equation} \label{HK}
K(x,x';T):=\mathrm{e}^{-T\Q}(x,x';T)\ .
\end{equation}
As customarily, in Eq.~\eqref{Gamma} an infinite additive constant is neglected; note also that the variable $T$ is often called the Schwinger proper time. The heat kernel operator \eqref{HK} can now be interpreted in quantum mechanical terms as the matrix elements $\braket{x|\mathrm{exp}(-T\Q)|x'}$ of the evolution operator $U(T)$ of an $N=0$ spinning particle, with a corresponding proper time $T$ for the evolution and a phase space Lagrangian
\begin{equation}
L=-ip_\mu\dot{x}^\mu+p^2+ V\ ,
\end{equation}
where the dot denotes $\dot{x}:=\frac{{\rm d}x(t)}{{\rm d}t}$.
Following the paradigm of the Worldline Formalism~\cite{Schubert:2001he}, one can represent the transition amplitude in terms of a path integral over the bosonic coordinates $x^\mu(t)$ in a first-quantized framework.
In this work we are going to follow two different approaches, depending on whether we absorb the spacetime integral into the path integral or not.
If we do not, then the effective action is computed from the coincidence point expression of
\begin{equation} \label{path}
K(y,z;T)=\int_{x(0)=y}^{x(T)=z} \mathcal{D}x \, \mathrm{e}^{-\int_{0}^{T} {\rm d}t \left( \frac{\dot{x}^2}{4}+ V(x) \right)}\ .
\end{equation}
This approach, which involves Dirichlet (D) boundary conditions (BC) on the worldline, is of fundamental importance when local quantities are to be analyzed, such as the energy-momentum tensor. If instead one is interested in global quantities, as is the case for the effective action in Eq.~\eqref{Gamma}, one can realize that the spacetime integral can be effectively incorporated into the boundary conditions of the path integration. 
A deeper discussion of this issue will be postponed to Sec.~\ref{sec:assisted_yukawa} and App.~\ref{appA}, where the so-called 
string-inspired (SI) BC will prove helpful.

\section{Resumming first and second derivatives of the potential: the heat kernel} \label{sec:resummation}
Having at disposal a Worldline representation for the relevant transition amplitude, cf. \eqref{path}, we are going to show that one can set up a perturbative expansion which already incorporates the information of all the invariants made up of the first and second derivatives of the potential. Since we are interested in working at the local level of the heat kernel, we will employ DBC on the worldline.

As a first step, we follow Ref.~\cite{Franchino-Vinas:2023wea} and Taylor expand the scalar potential about an arbitrary point $\tilde x$,
\begin{equation} \label{Taylor}
 V(x)= V(\tilde{x})+l^\mu \partial_\mu V(\tilde{x})+\frac12 l^\mu l^\nu \partial_{\mu\nu} V(\tilde{x}) + \dots\ ,
\end{equation}
where we define the distance $l^\mu:= (x-\tilde{x})^\mu$ and employ the following shorthand notation for the higher derivatives: $\partial_{\mu_1 \dots \mu_n}V:=\partial_{\mu_1}\dots \partial_{\mu_n}V$. Identifying the base point with one of the arguments of the heat kernel, $\tilde{x} \to y$, and performing the translation $x \to x' = x-y$ in the path integral as well (we will omit the prime henceforth), the expression for the heat kernel can be recast as
\begin{equation}\label{eq:hk_resummed_initial}
K(y,z;T)=\mathrm{e}^{-T V(y)} \, \int_{x(0)=0}^{x(T)=z-y} \mathcal{D}x \, \mathrm{e}^{-\int {\rm d}t \, x^\mu \partial_\mu V(y)} \, \mathrm{e}^{-S_{\mathrm{free}}-S_{\mathrm{int}}}\ ,
\end{equation}
where the quadratic part of the worldline action will be called the free action,
\begin{equation} \label{free}
S_{\mathrm{free}}:=\frac12 \int_{0}^{T} {\rm d}t \, x^\mu \left(\frac{1}{2} \delta_{\mu\nu}\, \overleftarrow{\partial_t}\,\overrightarrow{\partial_t} +\partial_{\mu\nu} V(y)\right)x^\nu\ ,
\end{equation}
while the higher-order terms will be included in the interacting action,
\begin{align} \label{int}
\begin{split}
S_{\mathrm{int}}:&=\int_{0}^{T} {\rm d}t \, \sum_{n=3}^{\infty} \frac{1}{n!}\, x^{\mu_1} \dots x^{\mu_n} \, \partial_{\mu_1 \dots \mu_n} V(y)\ ,
\\
&=: \int_0^T {\rm d}t\, L_{\rm int} (x(t))\ .
\end{split}
\end{align}
Note that all the tensorial quantities involving the background field $V$ and its derivatives are evaluated at the initial point $y$; as a consequence, their dependence on the worldline bosonic variables $x(t)$ is factored out and the expression~\eqref{eq:hk_resummed_initial} is thus readily usable for a perturbative computation in powers of the path. 

Using $S_{\rm free}$ as the base action for the expansions is convenient for two reasons. On the one hand, one only has to deal with a Gaussian path integral. On the other side, the perturbative expansion additionally corresponds to an expansion in the number of derivatives acting on a (single factor) $V$. 
Taking this into account, it proves useful to introduce an arbitrary external source $\eta_{\mu}(t)$ linearly coupled to the paths, as well as the corresponding generating functional of path n-point functions,
\begin{equation}
Z[\eta](y,z;T):= \int_{x(0)=0}^{x(T)=z-y} \mathcal{D}x \, \mathrm{e}^{-S_{\mathrm{free}}-\int_{0}^{T} {\rm d}t \, \eta_\mu x^\mu},
\end{equation}
which lead us to a master equation for the heat kernel:
\begin{equation} \label{HK2}
{K(y,z;T)=\left.\mathrm{e}^{-T V(y)} \, \mathrm{e}^{-\int_0^T {\rm d}t \, L_{\mathrm{int}}\left(-\frac{\delta}{\delta \eta(t)} \right)} \, Z[\eta](y,z;T)\right|_{\eta=\partial V(y)}}\ .
\end{equation}

\subsection{Generating functional} \label{sec2.1}
After an integration by parts in the worldline action, we can rewrite the generating functional as
\begin{equation} \label{gen}
Z[\eta](y,z;T)= \int_{x(0)=0}^{x(T)=z-y} \mathcal{D}x \, \mathrm{e}^{-\frac12 \int_{0}^{T} {\rm d}t \, \left(x^\mu \Delta_{\mu\nu} x^\nu+2\eta_\mu x^\mu \right)}\ ,
\end{equation}
where the action is defined in terms of a differential operator that acts on paths satisfying DBC on the interval $[0,T]$:
\begin{equation} \label{Delta}
{\Delta_{\mu\nu} (t,t')=-\frac12 \delta_{\mu\nu}\partial^2_t\delta(t-t')+2 \, \Omega^2_{\mu\nu}(y)\delta(t-t')}\ .
\end{equation}
To simplify the notation of the upcoming results, we have introduced $\Omega_{\mu\nu}$, which is related to the second derivative of the potential in the following way:
\begin{equation} \label{Omega}
{2 \, \Omega^2_{\mu\nu}(y):=\partial_{\mu\nu} V(y)}\ .
\end{equation}

Taking all this into account, the computation of $Z[\eta]$ reduces to the obtainment of both the inverse and the functional determinant of the operator $\Delta_{\mu\nu}$. The explicit computation goes as follows. First, it is convenient to recast the path integral in terms of the quantum fluctuations $\hat{s}^\mu(t)$ around the classical trajectory $x_{\mathrm{cl}}^\mu(t)$,
\begin{equation} \label{2.4}
x^\mu(t):=x_{\mathrm{cl}}^\mu(t)+\hat{s}^\mu(t)\ .
\end{equation}
Secondly, the classical trajectory $x_{\mathrm{cl}}^\mu(t)$ is defined as the solution to the equations of motion imposed by the free worldline action, i.e.\footnote{Here and in the subsequent differential equations for the worldline $x(t)$ we will employ an abuse of notation, such that whenever the differential operator appears as $\Delta$, we intend it striped-off of the Dirac delta that should appear according to Eq.~\eqref{eq:BC_D}. For instance, $\Delta \, x(t)=(-\tfrac12 \partial_t^2 +2\Omega^2) \, x(t)$.}
\begin{equation}\label{eq:BC_D}
\frac{\delta S_{\mathrm{free}}}{\delta x_\mu}=0 \quad \Longrightarrow \quad \Delta^{\mu\nu}x^{\mathrm{cl}}_\nu(t)=0\ ,
\end{equation}
satisfying the following boundary conditions 
\begin{equation} \label{bc}
x_{\mathrm{cl}}^\mu(0)=0\ , \quad x_{\mathrm{cl}}^\mu(T)=(z-y)^\mu\ .
\end{equation}
This fact, together with expression~\eqref{2.4}, implies that the fluctuations satisfy vanishing DBC, i.e.
\begin{equation} \label{DBC}
\hat{s}^\mu(0)=\hat{s}^\mu(T)=0\ .
\end{equation}

Coming back to the classical solution, it can be straightforwardly computed and is given by
\begin{equation} \label{hom}
x^\mu_{\mathrm{cl}}(t)=\left(\frac{\sinh{(2\Omega t)}}{\sinh{(2\Omega T)}}\right)^{\mu\nu}(z-y)_\nu\ ,
\end{equation}
where the tensorial character of $\Omega_{\mu\nu}$ has been explicitly shown. The resulting partition function thus reads 
\begin{equation} \label{2.9}
Z[\eta](y,z;T)= \, \mathrm{e}^{-S_{\mathrm{free}}[x_{\mathrm{cl}}]} \, \mathrm{e}^{-\int_{t} \eta_\mu x^\mu_{\mathrm{cl}}} \, \int_{\hat{s}(0)=0}^{\hat{s}(T)=0} \mathcal{D}\hat{s} \, \mathrm{e}^{-\frac12 \int_{t_1t_2} \hat{s}^\mu \Delta_{\mu\nu} \hat{s}^\nu-\int_t \eta_\mu \hat{s}^\mu}\ ,
\end{equation}
where we have introduced the following condensed notation for (multiple) integrals:\footnote{For instance, $\int_t \eta_\mu x^\mu=\int_0^T {\rm d}t \, \eta_\mu (t)x^\mu(t)$ and $\int_{t_1t_2} \hat{s} \Delta \hat{s}=\int_0^T {\rm d}t_1\int_0^T {\rm d}t_2 \, \hat{s}(t_1) \Delta(t_1,t_2) \hat{s}(t_2)$.} 
\begin{align}
 \int_{t_1t_2\cdots }:= \int_0^T {\rm d}t_1 \int_0^T {\rm d}t_2 \cdots. 
\end{align}

Performing the shift $\hat{s} \to \tilde{s}=\hat{s}+\Delta^{-1}\eta$ and a subsequent completion of squares one is lead to 
\begin{align} \label{Z}
\begin{split}
Z[\eta](y,z;T)&= \, \mathrm{e}^{-S_{\mathrm{free}}[x_{\mathrm{cl}}]} \, \mathrm{e}^{-\int_{t} \eta_\mu x^\mu_{\mathrm{cl}}} \, \mathrm{e}^{\frac12 \int_{t_1t_2} \eta^\mu \Delta^{-1}_{\mu\nu} \eta^\nu} \, \int_{\tilde{s}(0)=0}^{\tilde{s}(T)=0} \mathcal{D}\tilde{s} \, \mathrm{e}^{-\frac12 \int_{t_1t_2} \tilde{s}^\mu \Delta_{\mu\nu} \tilde{s}^\nu}
\\
&=\frac{\mathcal{C}_{\mathrm{DBC}} }{{\overline{\mathrm{Det}}^{\nicefrac{1}{2}}(\Delta)}} \, \mathrm{e}^{-S_{\mathrm{free}}[x_{\mathrm{cl}}] +\frac{1}{2} (S_1+ S_{\rm bos})}\ .
\end{split}
\end{align}
As a way to simplify the notation in Eq.~\eqref{Z}, we have introduced
actions for the linear and quadratic contributions in the external source,
\begin{align} \label{S1}
S_1[x_{\mathrm{cl}},\eta]:&=-2\int_t \, \eta_\mu(t) x^\mu_{\mathrm{cl}}(t)\ ,
\\
\label{Sbos}
S_{\mathrm{bos}}[\eta]:&=\int_{t_1 t_2} \, \eta^\mu(t_1) \Delta^{-1}_{\mu\nu}(t_1,t_2) \eta^\nu(t_2)\ ,
\end{align}
as well as the following short-hand notation 
for the quotient of functional determinants with respect to the free case: 
\begin{equation*}
 \overline{\mathrm{Det}}(\mathrm{A}) := \frac{\mathrm{Det}(\mathrm{A})}{\mathrm{Det}\left(-\tfrac12\delta_{\mu\nu}\partial^2_{\tau}\right)}\ .
\end{equation*}
Note that, in expression~\eqref{Z}, $\mathcal{C}_{\mathrm{DBC}}$ is a normalization to be determined, after the computation of the determinant and Green function of the $\Delta$ operator, from the well-known result for the Mehler kernel~\cite{Vinas:2014exa, Franchino-Vinas:2021bcl}.

In order to obtain an explicit expression for the generating functional, we begin by writing down the Green function corresponding to the operator $\Delta$:
\begin{align}\label{sGF}
\Delta^{-1}_{\mu\nu}(t,t')
&=\left[\frac{ \sinh(2\Omega t) \sinh\big(2\Omega (T-t')\big) - \Theta(t-t') \sinh(2\Omega T) \sinh\big(2\Omega (t-t')\big) }{\Omega \sinh(2\Omega T)}\right]_{\mu\nu}\ .
\end{align}
The details of its computation are given in App.~\ref{appA} and, for future reference, we report also its following integral
\begin{align}
E_{\mu\nu}(t):&=\int_{t'} \, \Delta^{-1}_{\mu\nu}(t,t') =\left[\frac{\sinh\left(\Omega (T-t)\right)\sinh\left(\Omega t\right)}{\Omega^2\cosh\left(\Omega T\right)} \right]_{\mu\nu}\ . \label{int1} 
\end{align}

Regarding the computation of the functional determinant, this can be carried out by employing the Gel’fand–Yaglom (GY) theorem~\cite{Gelfand:1959nq}, as generalized by Kirsten and McKane for a system of differential operators~\cite{Kirsten:2003py, Kirsten:2004qv}. 
As a matter of completion, the salient aspects of this procedure is given in App.~\ref{appB}. After a direct computation one gets the equality
\begin{equation} \label{DetS}
{\overline{\mathrm{Det}}(\Delta)=\mathrm{det}\left[ \frac{\sinh(2\Omega T)}{2\Omega T} \right]}\ .
\end{equation}
Finally, the computation of the normalization constant $\mathcal{C}_{\mathrm{DBC}}$ is straightforward; after replacing the results for $\Delta^{-1}$ and $\overline{\mathrm{Det}}(\Delta)$ into Eq.~\eqref{Z}, we obtain
\begin{equation} \label{Cs}
{\mathcal{C}_{\mathrm{DBC}} =\left( 4\pi T \right)^{-\nicefrac{D}{2}}}\ .
\end{equation}

\subsection{Resummed heat kernel} \label{sec2.2}
Recalling the formula~\eqref{HK2}, we use it in conjunction with the results of the precedent section to get the following representation for the diagonal of the heat kernel: 
\begin{equation} \label{R}
K(x,x;T)=\frac{\mathrm{e}^{-T V}}{(4\pi T)^{D/2} }\, \mathrm{det}^{-\nicefrac{1}{2}}\left(\frac{\sinh\left( 2\Omega T\right)}{2\Omega T}\right)\;\left. \mathrm{e}^{-\int_t \, L_{\mathrm{int}}\left(-\frac{\delta}{\delta \eta(t)} \right)} \, \mathrm{e}^{\frac12 S_{\mathrm{bos}}[\eta]}\right|_{\eta=\partial V}\ .
\end{equation}
The evaluation of this expression naturally proceeds by perturbatively expanding the exponent in powers of the operator-valued interacting terms contained in $L_{\mathrm{int}}$, i.e., in powers of functional derivatives $\frac{\delta}{\delta \eta}$. This can be readily used to obtain an expansion of the heat kernel up to a certain power in the proper time in an improved Schwinger--DeWitt expansion; in our case this is measured by considering number of derivatives acting on a single Yukawa potential. 

Explicitly, the action of the interacting exponential necessary to generate all the terms contributing up to order $T^5$ in the Schwinger--DeWitt expansion reads
\begin{align} 
\label{exp}
\begin{split} 
&\mathrm{e}^{-\int dt \, L_{\mathrm{int}}\left(-\frac{\delta}{\delta \eta(t)} \right)}\left.\mathrm{e}^{\frac12 S_{\mathrm{bos}}[\eta]}\right|_{\eta=\partial V}
\\
&=\mathrm{e}^{-\int dt \, \left( W_{(3)}(t)+W_{(4)}(t)+\dots\right)} \left.\mathrm{e}^{\frac12 S_{\mathrm{bos}}[\eta]}\right|_{\eta=\partial V}
\\
&=\bigg( 1-\int_t\, W_{(3)}(t)-\int_t\, W_{(4)}(t) -\int_t\, W_{(5)}(t) -\int_t\, W_{(6)}(t) -\int_t\, W_{(7)}(t)
\\
&\hspace{2.5cm}-\int_t\, W_{(8)}(t) 
+\frac12 \int_{t_1t_2} \; W_{(3)}(t_1)W_{(3)}(t_2)+\cdots \bigg) \, \left. \mathrm{e}^{\frac12 S_{\mathrm{bos}}[\eta]}\right|_{\eta=\partial V}\ ,
\end{split}
\end{align}
with the vertex-generating operators given by
\begin{align}
W_{(n)}(t)&=\frac{(-1)^{n}}{n!} \partial_{\mu_1 \mu_2 \cdots \mu_n} V\, \frac{\delta^n}{\delta\eta_{\mu_1}(t)\delta\eta_{\mu_2}(t)\cdots\delta\eta_{\mu_n}(t)}\ . \label{V3}
\end{align}
The evaluation of Eq.~\eqref{exp} can be easily performed in the Worldline Formalism; once inserted in \eqref{R} we get the result
\begin{align}
K(x,x;T)=\frac{\mathrm{e}^{-T V+\partial_\mu V\left[\frac{\Omega T-\tanh\left( \Omega T \right)}{4 \Omega^3}\right]^{\mu\nu} \partial_\nu V}}{(4\pi T)^{\nicefrac{D}{2}} \mathrm{det}^{\nicefrac{1}{2}}\left(\frac{\sinh\left( 2\Omega T\right)}{2\Omega T}\right)} \, \Sigma(x,x;T)\ ,
\end{align}
where $\Sigma(x,x;T)$ contains the information on higher derivatives of the potential, which can be written in terms of worldline diagrams; more explicitly, the contributions relevant at order $T^5$ are given by manageable worldline integrals,
\begin{align} \label{upsilon}
&\Sigma(x,x;T)=1+\frac12 \partial_{\mu\nu\rho} V \, \partial_\alpha V \int_t G^{\mu\nu}(t,t) E^{\alpha\rho}(t)\nonumber \\
&-\frac{1}{8}\partial_{\mu\nu\rho\lambda} V \int_t G^{\mu\nu}(t,t) G^{\rho\lambda}(t,t)
+\frac{1}{8}\partial_{\mu\nu\rho\lambda\tau} V \, \partial_\alpha V\int_t G^{\mu\nu}(t,t)G^{\rho\lambda}(t,t) E^{\alpha\tau}(t) \nonumber \\
&+\frac{1}{24} \partial_{\mu\nu\rho} V \partial_{\alpha\beta\gamma} V \int_{t_1 t_2}\bigg[ 3 G^{\mu\nu}(1,1) G^{\rho\alpha}(1,2)G^{\beta\gamma}(2,2) +2G^{\mu\alpha}(1,2)G^{\nu\beta}(1,2)G^{\rho\gamma}(1,2)\bigg] \nonumber \\
&-\frac{1}{48}\partial_{\mu\nu\rho\lambda\tau\theta} V \int_t G^{\mu\nu}(t,t) G^{\rho\lambda}(t,t)G^{\tau\theta}(t,t) \nonumber \\
&-\frac{1}{384}\partial_{\mu\nu\rho\lambda\tau\theta\alpha\beta} V \int_t G^{\mu\nu}(t,t) G^{\rho\lambda}(t,t)G^{\tau\theta}(t,t)G^{\alpha\beta}(t,t)+\dots\ ,
\end{align}
where we denoted $G(i,j):=G(t_i,t_j)$ in the arguments of the Green function $G^{\mu\nu}(t,t'):=\Delta^{-1}_{\mu\nu}(t,t')$ and $E^{\mu\nu}(t)$ has been introduced in Eq.~\eqref{int1}. Expanding in powers of the proper time,
\begin{equation}\label{eq:generalized_coeff}
\Sigma(x,x;T)=:\sum_{j=0}^{\infty} c_j(x,x) \, T^{j}\ ,
\end{equation}
we can read the first few coefficients in the improved Schwinger--DeWitt expansion: 
\begin{align}
c_0(x,x)&=1\ , \\
c_1(x,x)&=0\ , \\
c_2(x,x)&=0\ , \\
c_3(x,x)&=-\frac{1}{60} \partial^\mu{}_\mu{}^\nu{}_\nu V\ , \\
c_4(x,x)&=\frac{1}{30} \partial^{\mu}{}_{\mu\nu} V \, \partial^\nu V- \frac{1}{840} \partial^\mu{}_\mu{}^\nu{}_\nu{}^\rho{}_\rho V \ , \\
\begin{split}
c_5(x,x)&=\frac{17}{5040} \partial_\mu{}^\mu{}_\nu V \, \partial^\nu{}_\rho{}^\rho V +\frac{1}{840} \partial_{\mu\nu\rho} V \, \partial^{\mu\nu\rho} V\ , \\
&\phantom{=}+\frac{1}{280} \partial^{\rho\mu}{}_\mu{}^\nu{}_\nu V \, \partial_\rho V + \frac{1}{210} \partial^{\mu\nu} V \, \partial^\rho{}_{\rho\mu\nu} V -\frac{1}{15120}\partial^\mu{}_\mu{}^\nu{}_\nu{}^\rho{}_\rho{}^\lambda{}_\lambda V\ .
\end{split}
\end{align}
These coefficients are valid at the local level and, in particular, have been computed without the use of integration by parts; they are in perfect agreement with previous calculations---see Refs.~\cite{Franchino-Vinas:2023wea, Franchino-Vinas:2024wof} and references therein.

\section{Resummations for the assisted Yukawa interaction} \label{sec:assisted_yukawa}
As already highlighted in Ref.~\cite{Franchino-Vinas:2023wea}, the resummed expressions that we have obtained in the preceding section can be employed to analyze scenarios in which a strong field is involved.
In this realm, an exciting mechanism has been devised in Ref.~\cite{Schutzhold:2008pz}, where the rate of Schwinger pair creation was greatly enhanced by including an assisting, fastly varying field to the strong background one. 

In the following we are going to consider a perturbative approach to the assisted Yukawa pair production.
Using our worldline setup, we are going to study the imaginary part of the in-out effective action, which is closely related to the probability of pair creation in the weak-production limit (see below). This is an alternative path to those already considered for (S)QED, for which a scattering approach was considered in Ref.~\cite{Torgrimsson:2017pzs} and, as we will see, we will be able to obtain closed expressions without the need of employing saddle point approximations, which was necessary in Ref.~\cite{Torgrimsson:2018xdf}.

Let us then consider a potential of the form
\begin{equation}
 V(x)= V(x) +\epsilon \, \mathcal{V}(x) \ ,
\end{equation}
where $ V(x)$ is a strong field, while $\mathcal{V}(x)$ is a fastly varying field, whose strength is tuned by the parameter $\epsilon\ll 1$. 
Since we are interested in a global quantity (the effective action), we will employ a worldline model with string-inspired boundary conditions, which we describe in the following.

Departing from Eq.~\eqref{Gamma} for the effective action, we recall the interpretation of the heat kernel as a transition amplitude; this suggests that the spacetime integral of the diagonal of the heat kernel (the integrated heat kernel or heat kernel's trace) can be equivalently written as a path integral over periodic (P) trajectories,
\begin{align} \label{path3}
\begin{split}
{K}(T):&= \int {\rm d}^Dx \, K(x,x;T) 
\\
&=\oint \mathcal{D}x \, \mathrm{e}^{-\int_t \left( \frac{\dot{x}^2}{4}+ V(x) +\epsilon \, \mathcal{V}(x) \right)}\ . 
\end{split}
\end{align}
Afterwards, we introduce the loop's ``center of mass"
\begin{equation} \label{CdM}
\bar{x}^\mu=\frac{1}{T} \int_t \, x^\mu(t)\ .
\end{equation}
Once it is fixed, the path integration can be reobtained by subsequently integrating over all closed loops that share the center of mass; in other words, we decompose the closed worldlines as
\begin{equation}
\oint \mathcal{D}x=\int {\rm d}^D \bar{x} \, \oint \mathcal{D}s\ ,
\end{equation}
where the paths satisfy then the string-inspired boundary conditions:
\begin{equation} \label{s}
x^\mu(t)=\bar{x}^\mu+s^\mu(t) \quad \text{with} \quad \int_t \, s^\mu(t)=0\ .
\end{equation}
For the interested reader, a more general approach to boundary conditions in the worldline is explained in App.~\ref{appA}.

Coming back to the computation of the integrated heat kernel, for convenience we Taylor expand the strong potential $V$ about the center of mass variables:
\begin{equation}
 V(\bar{x}+s)= V(\bar{x})+s^\mu \partial_\mu V(\bar{x})+\frac12 s^\mu s^\nu \partial_{\mu\nu} V(\bar{x}) + \dots\ .
\end{equation}
This allows us to split once more the contributions of $V$ in the action~\eqref{path3} into an interacting part and a free one. Introducing the string-inspired generating functional
\begin{equation}
Z[\eta](\bar{x};T)=\oint \mathcal{D}s \, \mathrm{e}^{-\int_t \frac{\dot{s}^2}{4}+\frac12 s^\mu s^\nu \partial_{\mu\nu}V+s^\mu\eta_\mu+\epsilon \, \mathcal{V}(\bar{x}+s)}\ ,
\end{equation}
we recast the integrated heat kernel as
\begin{equation} \label{path4}
{K}(T)=\left. \int {\rm d}^D\bar{x} \, \mathrm{e}^{-TV}\, \mathrm{e}^{-\int_t L_{\mathrm{int}}\left( -\frac{\delta}{\delta \eta(t)} \right)} \, Z[\eta](\bar{x};T)\right|_{\eta=\partial V} .
\end{equation}
In the following, we are going to perform a perturbative expansion in the weak-field parameter $\epsilon$, 
\begin{equation}
Z=Z_0+\epsilon Z_1+\epsilon^2 Z_2 + \dots\ ;
\end{equation}
although each contribution $Z_i$ can be explicitly computed in the worldline, we are going to focus on the contributions up to order $\epsilon$, which will be enough to show the existence of the assisted effect.

\subsection{Zeroth order in $\mathcal{V}$}
At zeroth order, we need to compute
\begin{equation}
Z_0[\eta](\bar{x};T)=\oint \mathcal{D}s \, \mathrm{e}^{-\int_t \frac{\dot{s}^2}{4}+ s^\mu s^\nu \Omega_{\mu\nu}+s^\mu\eta_\mu}\ ,
\end{equation}
which is analogous to the path integral in Eq.~\eqref{2.9}, replacing the DBC with the SI ones. 
The relevant operator for this computation, $\Delta_{\rm SI}$, satisfies Eq.~\eqref{Delta} using $\Omega^2_{\mu\nu}(\bar{x}):=\tfrac12\partial_{\mu\nu} V(\bar{x})$ [instead of $\Omega^2_{\mu\nu}({y})$] and we should recall that its domain of definition is made of functions which obey the condition~\eqref{s}:
\begin{equation} \label{eq:Delta_SI}
{\Delta^{\mu\nu}_{\rm SI} (t,t')=-\frac12 \delta^{\mu\nu}\partial^2_t\delta(t-t')+2 \, \left(\Omega^2\right)^{\mu\nu}(\bar{x})\delta(t-t')}\ .
\end{equation}

After the replacement $s \to {s}=\tilde s+\Delta_{\rm SI}^{-1}\eta$, the generating functional can be brought into a Gaussian form, which can readily be integrated:
\begin{align} \label{3.44}
\begin{split}
Z_0[\eta](\bar{x};T)&= \mathrm{e}^{\frac12\int_{t_1t_2} \eta_\mu \left(\Delta^{-1}_{\rm SI}\right)^{\mu\nu}\eta^\nu} \, \oint \mathcal{D}\tilde{s} \, \mathrm{e}^{-\frac12\int_{t_1t_2} \tilde{s}_\mu\Delta^{\mu\nu}_{\rm SI} \tilde{s}_\nu}\\
&=\frac{\mathcal{C}_{\mathrm{SI}} \, \mathrm{e}^{\frac12 S^{\rm SI}_{\mathrm{bos}}[\eta]}}{{{\overline{\mathrm{Det}}}'}^{\nicefrac{1}{2}}(\Delta_{\mathrm{SI}})}\ .
\end{split}
\end{align}
This expression is made of three different elements. First, the exponential of the bosonic action, which is completely analogous to Eq.~\eqref{Sbos}, but for the fact that one should replace the operators with the appropriate SI ones. 
Second, the functional determinant of the $\Delta_{\rm SI}$ operator, which can be computed by generalizing the Gel’fand–Yaglom method to operators acting on a domain of periodic functions.\footnote{The prime over the determinant means that we are excluding the constant mode: in the SI BC this is automatically excluded, while in the free operator this is done by hand.} As explained in App.~\ref{appB}, the result is
\begin{equation}
\overline{\mathrm{Det}}'(\Delta_{\mathrm{SI}})=\mathrm{det}\left[ \frac{\sinh^2(\Omega T)}{\Omega^2T^2} \right]\ .
\end{equation}
Third and last, we have introduced the parameter $\mathcal{C}_{\mathrm{SI}}$; as in the Dirichlet case, it can be determined by appealing to a known result; once more we get
\begin{align} \label{CSI}
&\mathcal{C}_{\mathrm{SI}} = \left( 4\pi T \right)^{-\nicefrac{D}{2}}\ ,
\end{align}
which also agrees with the expected result for the small proper time expansion of the integrated heat kernel~ \cite{Vassilevich:2003xt}. Note that this normalization yields a well-defined limit for a vanishing $\Omega$.
Summing all these pieces, we get
\begin{align} \label{Z0}
\begin{split}
Z_0[\eta](\bar{x};T)&= \frac{ \mathrm{e}^{\frac12 S^{\rm SI}_{\mathrm{bos}}[\eta]}}{(4\pi T)^{\nicefrac{D}{2}}\mathrm{det}^{\nicefrac{1}{2}}\left(\frac{\sinh^2(\Omega T)}{\Omega^2T^2}\right)}\ .
\end{split}
\end{align}

\subsubsection{Non-assisted particle production}\label{sec:non-assisted}
As a warm up application of the preceding formulae, consider a strong quadratic potential, i.e. one whose higher derivatives can be neglected, $\partial^{n\geq 3}V= 0$.
This implies that the worldline interactions in Eq.~\eqref{path4} are switched off and, inserting the zeroth-order generating functional in the definition of the effective action, cf. Eq.~\eqref{Gamma}, we obtain the following effective Lagrangian:
\begin{equation} \label{3.45}
 \mathcal{L}_{\mathrm{eff}}[V]=-\frac12 \int_0^{\infty} \frac{{\rm d}T}{T} \, \frac{\mathrm{e}^{-TV}}{\left( 4\pi T \right)^{\nicefrac{D}{2}}} \, \mathrm{det}^{\nicefrac{1}{2}}\left( \frac{\Omega^2T^2}{\sinh^2(\Omega T)} \right)\ .
\end{equation}
The first aspect that deserves being mentioned is the fact that the proper time integral contains a singularity at $T=0$, indicating a UV divergence that needs to be removed by renormalization. It will eventually require the inclusion of counterterms in the effective Lagrangian, whose number will depend on the spacetime dimensionality $D$. However, as we will shortly see, in our strong field case the integral has developed further poles in the $T$-plane if we work in Minkowski spacetime. 

The determinant in the equation above can be evaluated from the knowledge of the eigenvalues of the real and symmetric matrix $\Omega^2$, defined in Eq.~\eqref{Omega}. These in turn depend on the invariants $\mathrm{det}(\Omega^2)$ and $\mathrm{tr}\left(\Omega^{2j}\right)$, for $j=1,2,.., D-1$. For $D\leq 4$ it is possible to obtain analytic expressions for the eigenvalues in terms of these invariants. However, the resulting expressions become cumbersome unless $D=2$, which serves as an interesting and illuminating special case. Therefore, in what follows we assume that the background potential depends just on two coordinates (e.g., $x_0$ and $x_1$). This results in a nontrivial $2\times 2$ block in the matrix $\Omega$, whose eigenvalues entirely determine the effective action. With a slight abuse of notation, we will continue to denote this block as $\Omega^2$.

The (real) eigenvalues of the $2\times 2$ block are given by
\begin{equation}
\lambda_\pm = -\mathfrak F\pm\sqrt{\mathfrak F^2 -\mathfrak G^2},
\end{equation}
where we have introduced the notation
\begin{equation}
\mathfrak F = -\frac{1}{2}\mathrm{tr}(\Omega^2)\, ,\quad \mathfrak G^2=\mathrm{det}(\Omega^2)\, .
\end{equation}
We have called these invariants $\mathfrak F$ and $\mathfrak G$, given that they play roles analogous to $\mathfrak{F}_{\rm QED}:=F_{\mu\nu} F^{\mu\nu}$ and $\mathfrak{G}_{\rm QED}:=\tilde F_{\mu\nu} F^{\mu\nu}$ in four-dimensional
quantum electrodynamics. 
Using these eigenvalues, the effective Lagrangian can be expressed as
\begin{equation} \label{Leff2D}
\mathcal{L}_{\mathrm{eff}}[V]=-\frac12 \int_0^{\infty} \frac{{\rm d}T}{T} \, \frac{\mathrm{e}^{-TV} }{\left( 4\pi T \right)^{\nicefrac{D}{2}}} \, \left( \frac{\sqrt{\lambda_+\lambda_-}T^2}{\sinh(\sqrt \lambda_+T)\sinh(\sqrt \lambda_-T)} \right)\ .
\end{equation}

Depending on the specific values of the eigenvalues, the effective Lagrangian can develop an imaginary part when rotated to Minkowski space. To patently see this in a particular case, let us choose $v\in\mathbb{R}$ and the strong potential in Euclidean time to be
\begin{equation}\label{eq:background_quadratic}
 V_0(x):=v^2 \, x_0^2+m^2\ ,
\end{equation}
where we have also included a mass term for convenience. As we will see, this is the analog of the constant electric field which was considered in Ref.~\cite{Franchino-Vinas:2023wea}. With this choice we have $\lambda_+=v^2$ and $\lambda_-=0$. Then, the Euclidean effective Lagrangian is given by
\begin{equation}
 \mathcal{L}_{\mathrm{eff}}[V_0]=-\frac{v}{2} \int_0^{\infty} {{ \rm d} T} \, \frac{\mathrm{e}^{-Tv^2\bar{x}_0^2-Tm^2}}{\left( 4\pi T \right)^{\nicefrac{D}{2}} {\sinh(v T)}} \ .
\end{equation}
Though the rotation to Minkowski space requires some care in the calculations, it can be done as follows. The Euclidean quantities $X_{\mathrm{E}}$ are related to the Minkowski ones by adding a $(-i)$ factor for every $0$th component which is involved. For example, for a vector we have $X^0_{\mathrm{E}}=-i X^0_{\mathrm{M}}$. Actually, we can keep track of the rotation by including a parameter $\phi$, so that $X_{\mathrm{E}}=\mathrm{e}^{i\phi} X_{\mathrm{M}}$; thus, expanding $\phi$ about $-\pi/2$ we can keep track of branch cuts and poles in the complex plane. The objects that we need to rotate include $v$ (which is to be counted as a derivative of the potential) and the effective action, 
which satisfies
\begin{equation}
 \Gamma_{\mathrm{M}}\Big\vert_{\text{Wick rotated}}=i \Gamma_{\mathrm{E}}\ .
\end{equation}
Taking this into account, the Minkowskian effective action at zeroth order reads
\begin{equation}\label{eq:Minkowski_nonassisted_EA}
 \Gamma^{(0)}_{\mathrm{M}}[V_0]=\frac{ \sqrt{\pi}}{2} \frac{\mathrm{Vol}_{\mathrm{(D-1)}}}{\left( 4\pi\right)^{\nicefrac{D}{2}}} \int_0^{\infty} \frac{{\rm d}T}{T^{\frac{D+1}{2}}} \, \frac{\mathrm{e}^{-m^2 T}}{\sin({v}_{\mathrm{M}} T + i 0)}\ ,
\end{equation}
where $\mathrm{Vol}_{\mathrm{(D-1)}}$ is the volume of the ($D-1$)-dimensional space and we have kept track of the rotation through the $+i 0$ contribution in the argument of the sine. 
Eq.~\eqref{eq:Minkowski_nonassisted_EA} shows nontrivial poles where the sine function vanishes, i.e. at\footnote{These are the values for which the differential operator $\Delta_{\mu\nu}$ with PBC develops zero modes in Minkowski spacetime. Although the $n=0$ one does not correspond to a zero eigenvalue, as can be extrapolated from the spectrum (see App.~\ref{app:det_pbc}), it is not surprising that it still represents a divergence, as it is well-known to be related to the renormalization of the theory.}
\begin{equation} \label{poles}
 T=\frac{\pi n}{{v}_{\mathrm{M}}}, \quad 0<n \in \mathbb{N} \ .
\end{equation}
On their turn, these poles generate an imaginary part in the effective action, as can be seen from the use of the residue theorem in the complex $T$ plane to perform the proper time integral,
\begin{align}\label{eq:im_ea_0}
\begin{split}
\operatorname{Im} \Gamma^{(0)}_{\mathrm{M}}
 &=\frac{\pi}{2} \frac{\mathrm{Vol}_{\mathrm{(D-1)}}}{\left( 2\pi\right)^{D}} \, v_{\mathrm{M}}^{\frac{D-1}{2}} \sum_{n=1}^\infty (-1)^{n+1} \frac{\mathrm{e}^{-\frac{m^2 \pi n }{v_{\mathrm{M}}}}}{n^{\frac{D+1}{2}}}
 \\
 &=-\frac{\pi}{2} \frac{\mathrm{Vol}_{\mathrm{(D-1)}}}{\left( 2\pi\right)^{D}} \, v_{\mathrm{M}}^{\frac{D-1}{2}} \operatorname{Li}_{\frac{D+1}{2}}\left(-\mathrm{e}^{-\frac{m^2 \pi }{v_{\mathrm{M}}}}\right)\ ,
 \end{split} 
\end{align}
where $\operatorname{Li}_s (\cdot)$ is the polylogarithm of order ${s}$.
As explained above, the physical interpretation of this result follows by identifying it with the vacuum persistence amplitude; indeed, the transition amplitude between the in and out vacua is $\langle 0_{\rm out} \vert 0_{\rm in}\rangle={\rm e}^{i \Gamma_\mathrm{M}}$, so that
\begin{align}
 \vert\langle 0_{\rm out} \vert 0_{\rm in}\rangle\vert^2 ={\rm e}^{- 2\operatorname{Im} \Gamma_\mathrm{M}} =:1-P_{\mathrm{pair}}\ .
\end{align}
In this formula we have defined the pair production probability, $P_{\mathrm{pair}}$; this definition matches our expectation, since it takes into account that the instability of the vacuum signals the appearance of states with a nonvanishing number of particles. As long as the latter are not largely populated, one can further approximate the probability of pair creation as
$P_{\mathrm{pair}} \approx 2 \, \mathrm{Im} \,\Gamma_{\mathrm{M}}$.

\subsubsection{Non-quadratic time-dependent potentials}
Let us consider an arbitrary time-dependent potential $V(t)$. 
Using the zeroth order result, to zeroth order also in the generalized heat kernel coefficients, we obtain what we will call the locally quadratic field approximation (LQFA),
\begin{align}
    \begin{split}
    \label{eq:eff_time_inho}
\Gamma_{\rm M}
&= -\frac{ \operatorname{Vol_{\rm D-1}}}{2 \sqrt{2}} \int_0^{\infty} \frac{{\rm d}T}{T} \, \int_{-\infty}^{\infty} {\rm d}t\frac{\mathrm{e}^{-T \left(V(t)+m^2\right)}}{\left( 4\pi T \right)^{\nicefrac{D}{2}}} \,  \frac{\sqrt{-V''(t)}T}{ \sinh\left(\sqrt{-V''(t)} T/\sqrt{2}\right)}, \
\end{split}
\end{align}
which heuristically is expected to be valid when the potentials are slowly varying; in this sense, it is equivalent to the locally constant field approximation  employed in QED~\cite{Dunne:2006st, Fedotov:2022ely}. 
For the computation of the imaginary part of the effective action, if the potential is positive, the relevant contributions will be those for which $V''$ is negative\footnote{For a negative potential, see the discussion at the end of Sec.~\ref{sec:spatial_pair}.}.
Importantly, one should be careful with the recipe employed to avoid the singularities that appear in the proper time integral. 
 To take care of this, we can follow as in the previous section; explained in other words, since our expressions are covariant, we can rotate the $00$th component of the metric, which we will call  $\eta_{00}$.
In order to be consistent with the $-\mathi \epsilon$ prescription, the rotation should be $\eta^{00}\to e^{-\mathi \pi}\eta^{00}$.
A further alternative is to consider the   heat kernel in Minkowski space and employ its imaginary proper time expansion~\cite{DeWitt:2003pm}. After a proper time Wick rotation, this implies a $T+\mathi \epsilon$ prescription that will be used in what follows (it will be especially useful for the space-dependent potentials).

Coming back to Eq.~\eqref{eq:eff_time_inho}, its imaginary part can be readily computed to be
\begin{align}
 \begin{split} &\operatorname{Im} \frac{\Gamma_{\rm M}}{\operatorname{Vol_{\rm D-1}}}   = \frac{\pi}{2} \, \sum_{r=1}^{\infty} (-1)^{r+1}  \int^{\infty}_{-\infty} {\rm d}t\,  \Theta(V''(t))
  \\
 &\hspace{4cm}\times \left(\frac{\sqrt{V''(t)} }{ 4 \sqrt{2} \pi^2  r }\right)^{\nicefrac{D}{2}}  \mathrm{exp}\left(\frac{-\sqrt{2}\pi r\left(V(t)+m^2\right)}{\sqrt{V''(t)}}\right)\, .
\end{split}
\end{align}
 We can go beyond LQFA  by including higher derivative terms, i.e. by multiplying the integrand in the equation above by using the improved Schwinger--DeWitt coefficients obtained in the SI approach (see App.~\ref{app:coeff_SI}). To lowest order,  we obtain
\begin{align}
\begin{split}\label{eq:im_time_inho}
  \operatorname{Im} \frac{\Gamma_{\rm M}}{\operatorname{Vol_{\rm D-1}}}=& \frac{\pi}{2} \,
\sum_{r=1}^{\infty} (-1)^{r+1} \int^{\infty}_{-\infty} {\rm d}t\,  \Theta\big(V''(t)\big) \left(\frac{\sqrt{V''(t)} }{ 4\sqrt{2}\pi^2 r}\right)^{\nicefrac{D}{2}}   \\
& \times \mathrm{exp}\left(\frac{-\sqrt{2} \pi r\left(V(t)+m^2\right)}{\sqrt{V''(t)}}\right) \left(1-\frac{2^{\nicefrac{3}{4}}\pi^3}{288} \frac{V^{(4)}(t)}{(V'')^{\nicefrac{3}{2}}}+\cdots\right).
\end{split}
\end{align}

As a particular case, consider an oscillatory potential; having in mind the analogy with the electromagnetic case, which in essence is $(F^2)_{\mu\nu}\to \Omega_{\mu\nu}$, a plausible form is
\begin{align}
    V(t)=\frac{V_0}{\omega^2} \sin^2(\omega t), \quad V_0>0,
\end{align}
which additionally reduces to the quadratic potential when $\omega\to 0$.
The Hessian matrix is defined in terms of the single element
\begin{align}
    V''(t)= 2{V_0} \cos(2\omega t);
\end{align}
inserting it  in Eq.~\eqref{eq:im_time_inho} and noting that the periodicity of the potential can be absorbed into the length $L_0=\frac{2\pi}{\omega}$,
we are led to the result
\begin{align}\label{eq:im_time_sine}
    \begin{split}
&\operatorname{Im} \frac{\Gamma_{\rm M}}{L_0\operatorname{Vol_{\rm D-1}}} 
= \, \sum_{r=1}^{\infty}  \frac{(-1)^{r+1}\pi}{4} \int^{\pi/4}_{-\pi/4} {\rm d}t \left(\frac{\sqrt{V_0 \cos(2 t)} }{ 4\pi^2 r }\right)^{\nicefrac{D}{2}}
\\
&\hspace{1cm} \times \mathrm{exp}\left(- \pi r \frac{ \sqrt{V_0}}{\omega^2}  \frac{\left( \sin^2 t+\gamma^2\right)}{ \sqrt{ \cos(2 t)}}\right)  \left(1+\frac{\pi^3}{36} \frac{\omega^2}{ \sqrt{V_0} } \frac{1}{\sqrt{\cos (2 t)}}+\cdots\right) \,,
\end{split}
\end{align}
where we have introduced Keldysh's adiabaticity parameter~\cite{Keldysh:1965ojf} :
\begin{align}
    \gamma:=\frac{m\omega}{\sqrt{V_0}}.
\end{align}
Note that the first correction in Eq.~\eqref{eq:im_time_sine} is positive and also reflects the large-potential character of our expansion, which is not in powers of $\gamma$, but in terms of the dimensionless, small parameter $\omega^2/\sqrt{V_0}$. For extreme fields,  one could  work out the time integral using Laplace's theory, which effectively leads to a parabolic approximation. Indeed, physically speaking, the expression~\eqref{eq:im_time_sine} can be seen as an improvement over simply approximating the harmonic function with a periodic array of quadratic potentials centered around its minima.
Additionally, in the small $\omega$ limit the  quadratic result is  reobtained; this can be seen either by using Laplace's theory in the previous equation or by rescaling in $\omega$,
\begin{align}
    \begin{split}\label{eq:im_time_oscillating}
\operatorname{Im} \frac{\Gamma_{\rm M}}{L_0 \operatorname{Vol_{\rm D-1}}}  
&\overset{\phantom{\omega\to 0}}{=}\frac{ \pi\omega }{4} \, \sum_{r=1}^{\infty} (-1)^{r+1}
\\
&\hspace{1cm}\times\int^{\pi/4\omega}_{-\pi/4\omega} {\rm d}t \left(\frac{\sqrt{V_0 \cos(2\omega t)} }{ 4\pi^2 r }\right)^{\nicefrac{D}{2}} \mathrm{exp}\left(\frac{-\pi r \left(\frac{V_0}{\omega^2} \sin^2(\omega t)+m^2\right)}{\sqrt{V_0 \cos(2\omega t)}}\right) 
\\
&\overset{\omega\to 0}{=}\frac{ \pi\omega}{4} \, \sum_{r=1}^{\infty} (-1)^{r+1} \left(\frac{v }{ 4\pi^2 r }\right)^{\nicefrac{D}{2}} \frac{1}{\sqrt{r v}}  \mathrm{exp}\left(\frac{-\pi r m^2}{v}\right),
\end{split}
\end{align}
where we have identified $\sqrt{V_0}\to v$.
Although we have not managed to obtain a closed analytic formula for Eq.~\eqref{eq:im_time_sine}, one can truncate the series at any desired order and readily numerically perform the integrations, at least for massive fields, for which a decreasing exponential behaviour is guaranteed.

As a last comment, we would like to emphasize once more that we expect Eq.~\eqref{eq:im_time_sine} to be a good approximation for small frequency $\omega$. If an all-scalar version of the worldline instantons technique could be developed, it could be used to  check all these (and the following section's) results and complement them with expressions encompassing all the derivatives contributions (but just the first large mass one).


\subsubsection{Spatial particle creation}\label{sec:spatial_pair}
The framework is slightly more involved when  considering the spatial case. For simplicity we will analyze the case depending on just one variable, say $x_3$, so that
\begin{align}
    \begin{split}
    \label{eq:eff_sapce_inho}
\Gamma_{\rm M}
&= -\frac{\mathi \operatorname{Vol_{\rm D-1}}}{2\sqrt{2}} \int_0^{+\infty} \frac{{\rm d}T}{T} \, \int_{-\infty}^{\infty} {\rm d}x_3\frac{\mathrm{e}^{-T \left(V(x_3)+m^2\right)}}{\left( 4\pi T \right)^{\nicefrac{D}{2}}} \,  \frac{\sqrt{V''(x_3)}T}{ \sinh\left(\sqrt{V''(x_3)} T/\sqrt{2}\right)}. \
\end{split}
\end{align}
In this setup, for weak fields instabilities arise when $\partial^2_3 V$ becomes negative; the difficulty resides in the fact that the instabilities are already present  at the level of the Euclidean effective action and, thus, the imaginary proper time prescription described in the previous section is required in order to obtain a sensible answer. 

To simplify further the discussion, let us consider the quadratic case 
\begin{align}
    V(x_3)=v^2 \, x_3^2+m^2,
\end{align}
leaving the sign of $v^2$ for the moment undetermined; note that the situation with the inverted quadratic potential mimics what happens in the electromagnetic homogeneous field in the spatial gauge. 
Then, following the lines of the previous section we arrive at
\begin{align}\label{eq:eff_spatial_har}
    \frac{\Gamma_\mathrm{M}[V]}{\operatorname{Vol}_{\rm D-1}}=-\frac{\mathi}{2} \int_0^{\infty} \frac{{\rm d} T}{T} \, \int^{\infty}_{-\infty}{\rm d}x_3\frac{\mathrm{e}^{-T(v^2{x}_3^2+m^2)}}{\left( 4\pi T \right)^{\nicefrac{D}{2}}} \,  \frac{v T}{\sinh(v T)}\ .
\end{align}
At this point we have to consider the two possible alternatives separately. If $v^2>0$, all the integrals in Eq.~\eqref{eq:eff_spatial_har} are well defined and are real, i.e. there is no  creation of pairs. This is of course logical: since the potential is static, no energetic source is available to create the pairs. If instead $v^2<0$, the sinh in the denominator becomes a sine and singularities appear. 
However, the integral in $x_3$ becomes ill-defined; this can be seen as a consequence of the fact that the effective mass $V+m^2$ is negative in an infinite region. 
A way to circumvent this problem in this simple case is to perform first the integral in $x_3$ (with $v^2>0$) and afterwards appeal to an analytical continuation in\footnote{One should employ the rule of thumb that the rotation should not interfere with the $\mathi \epsilon$ Feynman rule.} $v$. The result obtained with this prescription agrees with the one using the imaginary proper time, and gives the same pair creation probability as for the quadratic time dependence. This agreement resembles the situation in the electromagnetic homogeneous case. 

As a last example, consider another spatial, now oscillating background, 
\begin{align}\label{eq:spatial_harmonic}
    V(x_3)=\frac{V_0}{\omega^2} \sin^2(\omega x_3). 
\end{align}
If $V_0>0$, then the only change with respect to the time-dependent field is in the integration region over $x_3$, which is effectively shifted to the complementary interval in the period of\footnote{In this case, $\rm Vol_{D}=\mathit{L}_3 Vol_{D-1}$, where $L_3$ denotes a spatial cutoff.} $\cos (2\omega x_3)$; after shifting the integral one gets:
\begin{align}\label{eq:im_space_sine}
    \begin{split}
\operatorname{Im} \frac{\Gamma_{\rm M}}{\operatorname{Vol_{\rm D}}} 
=& \, \sum_{r=1}^{\infty}  \frac{(-1)^{r+1} \pi}{4}
\\
&
\times \int^{\pi/4}_{-\pi/4} {\rm d}x_3 \left(\frac{\sqrt{V_0 \cos(2x_3)} }{ 4\pi^2 r }\right)^{\nicefrac{D}{2}}  \mathrm{exp}\left(\frac{-\pi r \left(\frac{\sqrt{V_0}}{\omega^2} \cos^2(x_3)+\frac{m^2}{\sqrt{V_0}}\right)}{\sqrt{ \cos(2x_3)}}\right)
\,. 
\end{split}
\end{align}
Since our expansion is valid for strong potentials, the integral in the coordinate $x_3$ can be done using Laplace's theory; contrary to the time-dependent scenario,  one notices that $\cos^2 x_3\geq 1/2$ in the integration region, and thus the integral is exponentially suppressed in the strong field expansion. This is consistent with the fact that for a positive potential one does not expect particle creation to take place. This can actually be discussed in more general terms: whenever the  potential is strictly positive and its second derivative is negative, an exponential suppression is present in the imaginary part of the effective action in the LQFA. 

Taking into account our discussion for the quadratic case, one would expect that the harmonic potential~\eqref{eq:spatial_harmonic} with $V_0<0$ should resemble the most an electromagnetic space-dependent field. However, in this scenario there is a caveat: for strong (negative) fields the integral in Eq.~\eqref{eq:eff_sapce_inho} develops infrared singularities. These are related to the divergent behaviour for large $T$ and could imply a further source of imaginary contributions in the effective action; they are of course not present for QED, since gauge invariance precludes the presence of the gauge potential in the effective action. One could still appeal for example to analytical continuations to analyze the infrared singularities, as done for the quadratic field; however, for the oscillatory potential the integrals can not all be done explicitly and the same can be said for the most frequently analyzed pulses in the literature. 

For the sake of completeness, let us briefly discuss an analogue expansion  in QED; a thorough analysis will be  left for the future. In such a case, the first derivative contribution, i.e. the equivalent of the $\omega^2/\sqrt{V_0}$ term in Eq. \eqref{eq:im_time_sine}, is given by the $c_3^{\rm QED}$ coefficient. For a background field which is exclusively electric ($F_{0i}=E_i(x)$ and $F_{ij}=0$), it can be shown that 
\begin{align}
    c_3^{\rm QED}=\frac{2}{5}\left[ \partial_j E_i  \partial^j E^i - \partial_0 E_i \partial^0 E^i \right],
\end{align}
cf. the $\bar{o}_3$ coefficient in Ref.~\cite{Navarro-Salas:2020oew}.
This expression clearly displays a sign flip when shuffling from purely spatial to time-dependent potentials, meaning that the pair production probability is enhanced in the latter case, while it decreases for the former. This seems to be in agreement with the threshold which has been observed using worldline instantons for potentials depending on spatial coordinates~\cite{Dunne:2006st}.

\subsection{First and higher orders in $\mathcal{V}$} \label{sec3.4.2}
Before analyzing the first-order contribution in $\epsilon$, we will sketch how the computation to any order can be done. 
At $n$th order in the assisting potential, the contribution to the generating functional is
\begin{align}\label{eq:Zn}
\begin{split}
&Z_n[\eta](\bar{x};T)=(-1)^n\oint \mathcal{D}s \, \int_{t_1 t_2\cdots}\, \mathrm{e}^{-\int_t \frac{\dot{s}^2}{4}+ s^\mu s^\nu \Omega_{\mu\nu}+s^\mu\eta_\mu} \, \mathcal{V}(\bar{x}+s(t_1))\mathcal{V}(\bar{x}+s(t_2))\cdots
\\
&=(-1)^n\int \left[\prod_{j=1}^n \hat{\rm d}q_j \tilde{\mathcal{V}}(q_1) \right] 
\int_{t_1\cdots}\, \oint \mathcal{D}s \, \mathrm{e}^{-i \sum_{l=1}^n q_l \cdot \bar{x}-\int_t \frac{\dot{s}^2}{4}+s^\mu s^\nu \Omega_{\mu\nu}+s^\mu\eta_\mu+is_\mu\sum_{i=1}^n q_i^\mu \delta(\tau-t_i)} \, \ ,
\end{split}
\end{align}
where we have Fourier transformed $\mathcal{V}$ as
\begin{equation} \label{FT}
\mathcal{V}(x)=\int\hat{\rm d}q \, \mathrm{e}^{-i q \cdot x } \, \tilde{\mathcal{V}}(q)\ 
\end{equation}
and used the short-hand notation $\hat{\rm d}q:=\frac{{\rm d}^Dq}{(2\pi)^D}$. 
The expression \eqref{eq:Zn} is still manageable, since introducing explicitly the classical solution to the equations of motion of the worldline,
\begin{equation} \label{eq:}
s^\mu \to s^\mu_{\mathrm{cl}} + \hat{s}^\mu\ ,
\end{equation}
the exponent becomes quadratic in the quantum fluctuations.
A direct calculation shows that we need to find the solution of
\begin{equation}
(\Delta_{\rm SI})_{\mu\nu} \, s_{\mathrm{cl}}^\nu(\tau)=-\eta_\mu(\tau)-\mathfrak{d}_\mu(\tau)\ ,
\end{equation}
where $\mathfrak{d}$ contains all the momentum inhomogeneities:
\begin{align}
 \mathfrak{d}^\mu(\tau):= i \sum_{i=1}^n q^\mu_i\delta(\tau-t_i)\ .
\end{align}
This can be simply solved by considering the Green function of the operator, $\mathcal{G}:=\Delta_{\rm SI}^{-1}$, which is computed in App.~\ref{appA}:
\begin{align}
\mathcal{G}^{\mu\nu} (t,t')
&= \left[-\frac{1}{2T\Omega^2}+ \frac{\cosh{(\Omega(T+2t-2t'))}}{2\Omega\sinh{(\Omega T)}}- \Theta(t-t') \, \frac{\sinh{(2\Omega(t-t'))}}{\Omega} \right]^{\mu\nu}\!\! ,
\end{align}
so that
\begin{align}
s^\mu_{\mathrm{cl}}(\tau)&=-i \sum_{i=1}^n \mathcal{G}^{\mu}{}_{\nu}(\tau,t_i) q^\nu_{i} -\int_{t'} \, \mathcal{G}^{\mu\nu}(\tau,t')\eta_\nu(t')\label{eq:classical_solution_SI}\ .
\end{align}
After replacing Eq.~\eqref{eq:classical_solution_SI} in the expression for $Z_n$, we obtain 
\begin{align}
&Z_n[\eta](\bar{x};T) \nonumber
\\
&=(-1)^n\int \left[\prod_{j=1}^n \hat{\rm d}q_j \tilde{\mathcal{V}}(q_1) \right]\, \mathrm{e}^{-i \sum_{l=1}^n q_l \cdot \bar{x}} 
\int_{t_1t_2\cdots} \oint \mathcal{D}s \, \mathrm{e}^{-\frac12 \int_{t t'} \hat{s}^\mu \Delta_{\mu\nu} \hat{s}^\nu
+\frac12 \int_{t t'} \left( \eta +\mathfrak{d}\right)_\mu \mathcal{G}^{\mu\nu} \left( \eta +\mathfrak{d}\right)_\nu} \, \nonumber
\\
&=(-1)^n\frac{\mathrm{e}^{-TV}}{\left( 4\pi T \right)^{\nicefrac{D}{2}}} \, \mathrm{det}^{\nicefrac{1}{2}}\left( \frac{\Omega^2T^2}{\sinh^2(\Omega T)} \right) \nonumber
\\
&\hspace{2cm}\times
 \int \left[\prod_{j=1}^n \hat{\rm d}q_j \tilde{\mathcal{V}}(q_1) \right]\, \mathrm{e}^{-i \sum_{l=1}^n q_l \cdot \bar{x}} 
\int_{t_1t_2\cdots}\, \mathrm{e}^{\frac12 \int_{t t'} \left( \eta +\mathfrak{d}\right)_\mu \mathcal{G}^{\mu\nu} \left( \eta +\mathfrak{d}\right)_\nu} \ .
\end{align}
If we are just interested in the contribution without worldline interactions, this formula can be further simplified to
\begin{align}\label{eq:Zn_final}
\begin{split}
Z_n[\eta](\bar{x};T)
&=(-1)^n\frac{\mathrm{e}^{-TV}}{\left( 4\pi T \right)^{\nicefrac{D}{2}}} \, \mathrm{det}^{\nicefrac{1}{2}}\left( \frac{\Omega^2T^2}{\sinh^2(\Omega T)} \right)
\\
&\hspace{2cm}\times\int \left[\prod_{j=1}^n \hat{\rm d}q_j \tilde{\mathcal{V}}(q_1) \right]\, \mathrm{e}^{-i \sum_{l=1}^n q_l \cdot \bar{x}} 
\int_{t_1t_2\cdots}\, \mathrm{e}^{\frac12 \int_{t t'} \mathfrak{d}_\mu \mathcal{G}^{\mu\nu} \mathfrak{d}_\nu} \ ,
\end{split}
\end{align}
given that the SI Green function satisfies
\begin{align}
 \int_{t'} \mathcal{G}^{\mu\nu}(t,t') \partial_\mu V(\bar{x})=0\,\ .
\end{align}
It is important to note that, in Eq.~\eqref{eq:Zn_final}, an integral in $\bar{x}$ is not simply going to give a Dirac delta, which would imply momentum conservation, because of our nonperturbative approach. Indeed, there are also implicit dependences on $\bar{x}$ through $V$ and $\Omega$; the former will be a crucial difference with respect to the Abelian gauge field case.

Let us now focus on the first order contribution in the assisting potential. 
The relevant computation reads
\begin{align} \label{3.57}
\begin{split}
Z_1[\eta](\bar{x};T)=&- \frac{\mathrm{e}^{-TV}}{\left( 4\pi T \right)^{\nicefrac{D}{2}}} \, \mathrm{det}^{\nicefrac{1}{2}}\left( \frac{\Omega^2T^2}{\sinh^2(\Omega T)} \right)
\\
&\hspace{1.5cm}\times\int \hat{\rm d}q\,\mathrm{e}^{-i q \cdot \bar{x}-\frac12 q_\mu \mathcal{E}^{\mu\nu}q_\nu} \, \tilde{\mathcal{V}}(q) \int_{t_1} \, \mathrm{e}^{\frac12 S_{\mathrm{bos}}[\eta]+i q \int_t \mathcal{G}(t,t_1) \, \eta(t)}\ ,
\end{split}
\end{align}
with
\begin{equation} \label{EE}
\mathcal{E}^{\mu\nu}:=\mathcal{G}^{\mu\nu}(t,t)=\left[ \frac{\Omega T \coth{(\Omega T)}-1}{2\Omega^2 T}\right]^{\mu\nu}\ .
\end{equation}
Once these results are replaced in Eq.~\eqref{path4}, we can perturbatively recast the trace of the heat kernel as 
\begin{align}\label{eq:trace_HK_assisted_1order}
\begin{split} 
\tilde{K}(T)=-\frac{1}{\left( 4\pi T \right)^{\nicefrac{D}{2}}} &\int {\rm d}^D\bar{x} \; \mathrm{e}^{-TV} \, \mathrm{det}^{\nicefrac{1}{2}}\left( \frac{\Omega^2T^2}{\sinh^2(\Omega T)} \right)
\\
\times&\int \hat{\rm d}q\,\mathrm{e}^{-i q \cdot \bar{x}-\frac12 q_\mu \mathcal{E}^{\mu\nu}q_\nu} \, \tilde{\mathcal{V}}(q) \int_{t_1}
\, \big(1+ \Sigma^{(1)}(\bar{x};T,t_1)\big)\ ,
\end{split}
\end{align}
where higher-derivatives contributions in the Yukawa potential are encoded in $\Sigma^{(1)}(\bar{x};T,t_1)$.
If we split them according to the number of momenta they contain, 
\begin{equation} 
\Sigma^{(1)}(\bar{x};T,t_1)=\Sigma^{(1)}_0(\bar{x};T)+\Sigma^{(1)}_1(\bar{x};T,t_1)+\cdots \ ,
\end{equation}
we get the first terms
\begin{align}
\Sigma^{(1)}_0(\bar{x};T)&=-\frac{T}{8}\partial_{\mu\nu\rho\lambda} V \, \mathcal{E}^{\mu\nu} \mathcal{E}^{\rho\lambda} -\frac{T}{48}\partial_{\mu\nu\rho\lambda\tau\theta} V \, \mathcal{E}^{\mu\nu} \mathcal{E}^{\rho\lambda} \mathcal{E}^{\tau\theta}
 \nonumber \\
&\hspace{1cm}-\frac{T}{384}\partial_{\mu\nu\rho\lambda\tau\theta\alpha\beta} V \, \mathcal{E}^{\mu\nu} \mathcal{E}^{\rho\lambda} \mathcal{E}^{\tau\theta} \mathcal{E}^{\alpha\beta}\nonumber
\\
&\hspace{1cm}+\frac{1}{12} \partial_{\mu\nu\rho} V \partial_{\alpha\beta\gamma} V \, \int_{a b} \mathcal{G}^{\mu\alpha}(a,b)\mathcal{G}^{\nu\beta}(a,b)\mathcal{G}^{\rho\gamma}(a,b) 
+\dots\ ,
\\
\begin{split}
\Sigma^{(1)}_1(\bar{x};T,t_1)&=\frac{i}{2} \,\partial_{\mu\nu\rho} V \, q_{\alpha} \int_t \mathcal{G}^{\mu\nu}(t,t) \mathcal{G}^{\rho\alpha}(t,t_1) \\
&\hspace{1cm}+\frac{i}{6} \, \partial_{\mu\nu\rho\lambda\tau} V \, q_{\alpha}\int_t \mathcal{G}^{\mu\nu}(t,t)G^{\rho\lambda}(t,t) \mathcal{G}^{\tau\alpha}(t,t_1)+\dots\ ,
\end{split}
\end{align}
where we reported only the terms contributing to order $T^5$ once expanded in powers of $T$.
These results can be compared with the known expansion for integrated heat kernels, the difference with the local results in Sec.~\ref{sec:resummation} consisting just in boundary contributions.

\subsubsection{Gauss assisted pair creation at first order}\label{sec:assisted}
The results from the previous section can be readily used to analyze the effect of assisted pair creation. 
For simplicity, let us consider a minimal model in which the background, assisted field is quadratic in the coordinates, so that the integrated heat kernel simplifies to
\begin{align}\label{eq:HK_gaussian_assisted}
\begin{split} 
\tilde{K}(T)&=-\frac{T}{\left( 4\pi T \right)^{\nicefrac{D}{2}}} \int {\rm d}^D\bar{x} \; \mathrm{e}^{-TV} \, \mathrm{det}^{\nicefrac{1}{2}}\left( \frac{\Omega^2T^2}{\sinh^2(\Omega T)} \right)
\int \hat{\rm d}q\,\mathrm{e}^{-i q \cdot \bar{x}-\frac12 q_\mu \mathcal{E}^{\mu\nu}q_\nu} \, \tilde{\mathcal{V}}(q) \ .
\end{split}
\end{align}
Without losing generality, we can further restrict the potential to be only time-dependent, as we have done in Eq.~\eqref{eq:background_quadratic}:
\begin{equation}
 V_0(x)=v^2 \, x_0^2+m^2\ , \quad v=\text{const}\ .
\end{equation}
As assisting field we propose instead a Gaussian (infinitely wide) pulse, which is frequently employed in the description of experimental setups: 
\begin{equation}
 \mathcal{V}(x)=\frac{\omega^2}{\pi} \, \mathrm{e}^{-\omega^2 x_0^2}\ .
\end{equation}
Replacing these profiles for the fields into Eq.~\eqref{eq:HK_gaussian_assisted} and performing the integral in $\bar{x}_0$ and in the energy $q_0$, we get a compact expression for the trace of the heat kernel:
\begin{align}
\begin{split}
 \tilde{K}(T)
 &=-\frac{\mathrm{Vol}_{\mathrm{(D-1)}} }{(4\pi T)^{\nicefrac{D}{2}}} \frac{ \omega T^{\nicefrac{3}{2}} }{\sinh{\left( vT \right)}} \, \mathrm{e}^{-m^2T}\,\left( \frac{1}{\omega^2}+\frac{\coth{(vT)}}{v} \right)^{-\nicefrac{1}{2}}\ .
\end{split}
\end{align}
Using this expression, the Euclidean effective action is immediately obtained. As we have already performed at zeroth order in the assisting field, in the present case we can analyze the Minkowskian scenario by performing a Wick rotation, the only difference being that the frequency $\omega$ shall be also included among the quantities that have to be Wick-rotated. 
Keeping track once more of the necessary regularizations resulting from the rotation, we get
\begin{align}
\begin{split}
 \Gamma_\mathrm{M}^{(1)}
 &=-\frac{\mathrm{Vol}_{\mathrm{(D-1)}} \epsilon \omega_\mathrm{M}}{2(4\pi)^{\nicefrac{D}{2}}} \int_0^\infty \frac{{\rm d}T}{T^{\nicefrac{(D-1)}{2}}}
 \frac{\mathrm{e}^{-m^2T}}{\sin{\left( v_\mathrm{M} T + i 0\right)}} \, \left( \frac{1}{\omega_\mathrm{M}^2}+\frac{\cot{( v_\mathrm{M}T)}}{v_\mathrm{M}} -i0 \right)^{-\nicefrac{1}{2}}.
\end{split}
\end{align}
The analytic structure of the integrand in this expression is severely influenced by the square root factor. On the one hand, notice that the zeros of the sine are not poles of the integrand, given that the cotangent inside the square root partially compensates the divergence, rendering it integrable. On the other hand, its argument vanishes periodically in $T$, generating an infinite number of branch cuts in the complex $T$ plane, the $-i0$ resulting from the Wick rotation telling us on which side of the cut the integral shall be performed. 
To simplify the discussion, let us use $v_\mathrm{M}$ as scale and define the following complete set of dimensionless parameters:
\begin{align}
\bar{\omega}:&= \frac{\omega_{\rm M}}{v_\mathrm{M}^{\nicefrac{1}{2}}}\ ,
 \\
\bar{m}:&=\frac{m}{v_\mathrm{M}} \ ,
\\
\bar{T}^*:&= v_\mathrm{M} T^*\ .
\end{align}
Then, the cuts from the cotangent give rise to an imaginary part in the effective action, whose analytic expression can be straightforwardly found to be
\begin{align}
 \begin{split}\label{eq:im_ea_assisted}
\operatorname{Im} \Gamma_\mathrm{M}^{(1)}
 &=\frac{\pi\mathrm{Vol}_{\mathrm{(D-1)}} }{2(2\pi)^{{D}{}} } {v_\mathrm{M}}^{\nicefrac{(D-1)}{2}} \epsilon \mathrm{e}^{-\bar{m}^2 \bar{T}^* } \bar{\omega}
 \int_{0}^{\pi- \bar{T}^*} {\rm d}T \frac{\mathrm{e}^{-\bar{m}^2T}}{\sin{\left( T + \bar{T}^* \right)}}
\\
&\hspace{1cm}\times \, \left| \bar{\omega}^{-2}+\cot(T+\bar{T}^*)\right|^{-\nicefrac{1}{2}} \Phi\left(-\mathrm{e}^{-\pi \bar{m}^2}, \frac{D-1}{2}, \frac{T+\bar{T}^*}{\pi}\right),
 \end{split}
 \end{align}
where $T^*$ is the first positive root of $\bar{\omega}^{-2}+\cot{( \bar{T}^*)}=0$ and
$\Phi(\cdot,\cdot, \cdot)$ is the Lerch transcendent function,
\begin{align}
 \Phi(z,s,a):=\sum_{n=0}^{\infty} \frac{z^n}{(a+n)^s}.
\end{align}

Importantly, $T^*$ always belongs to the interval $T^*\in [\pi/2, \pi)$, since we consider $v_\mathrm{M},\omega_{\mathrm M}>0$; this implies that the exponential decay of the pair creation effect can be greatly softened by the assisted field.
Such a softening can be readily observed by inspecting the ratio of the pair creation probability at first assisted order, $P^{(1)}$, with respect to the non-assisted case, $P^{(0)}$. In the left panel of Fig.~\ref{fig:pair_creation}, one can observe a density plot of $\log\left(\frac{P^{(1)}}{P^{(0)}}\right)$ as a function of $\bar m$ and $\bar\omega$; the integral in the proper time $T$ has been performed numerically, while the value $\epsilon=10^{-3}$ has been chosen so that in the depicted region we roughly satisfy the small potential criterium $\epsilon\bar\omega^2\ll 1$.

As expected, the ratio of probabilities fastly increases as a function of the rescaled frequency $\bar\omega$. 
Indeed, the higher the frequency $\omega$ the larger the available energy to catalyze the production of pair creation. One can also see that, for the depicted values of the parameters, the ratio also increases with the mass; the reason is that the transseries structure in Eq.~\eqref{eq:im_ea_assisted} is shifted with respect to that in Eq.~\eqref{eq:im_ea_0}, thanks to the exponential prefactor ${\mathrm{e}}^{-\bar m ^2 \bar T^*}$.

\begin{figure}[h]
\begin{center}
 \begin{minipage}{0.48\textwidth}
 \includegraphics[width=1.0\textwidth,height=0.8\textwidth]{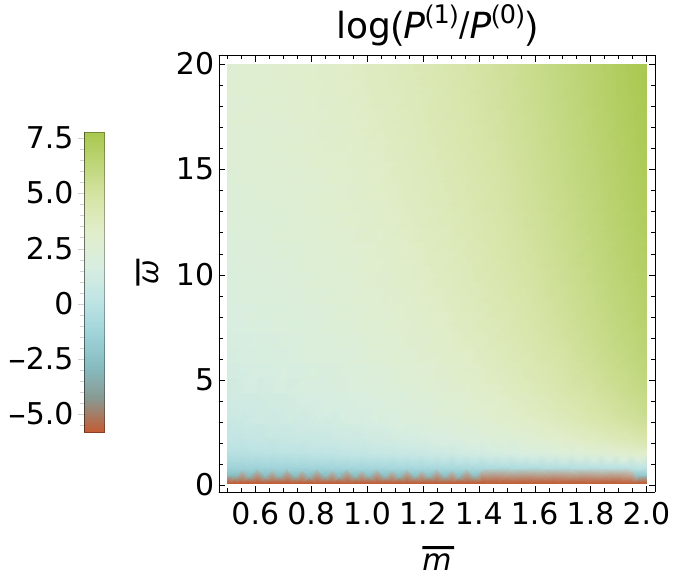}
 \end{minipage}
 \hspace{0.02\textwidth}
 \begin{minipage}{0.48\textwidth}
 \vspace{0.4cm}\includegraphics[width=0.95\textwidth]{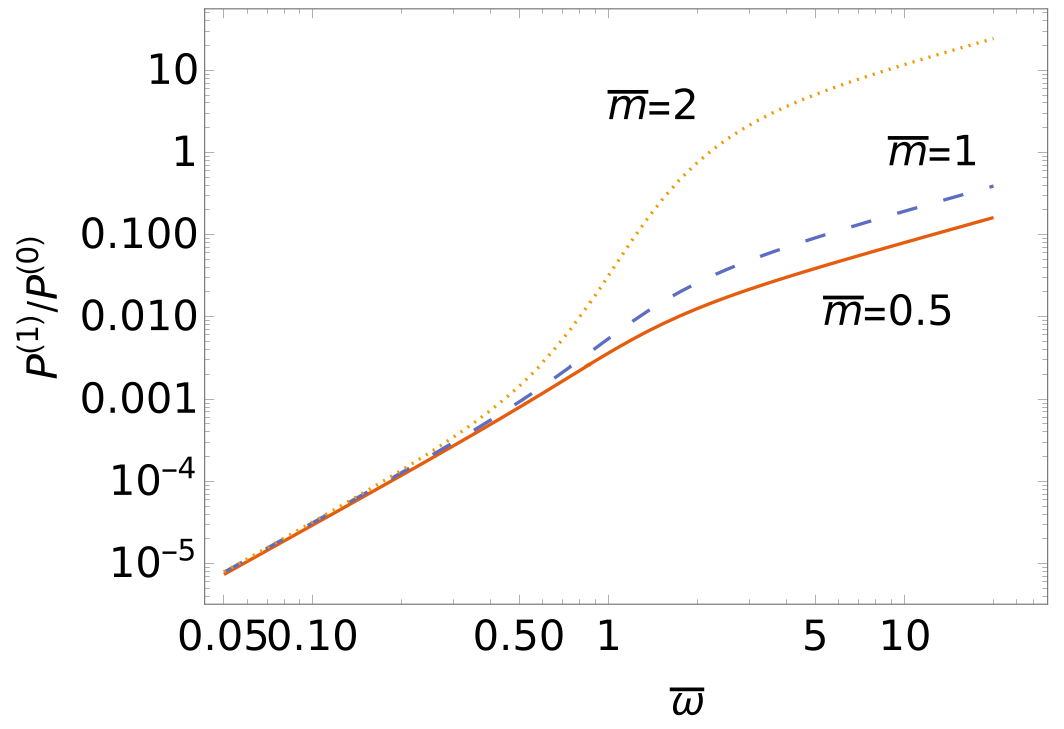}
 \end{minipage}
 \caption{The left panel is a density plot of $\log\left(\frac{P^{(1)}}{P^{(0)}}\right)$ as a function of $\bar m$ and $\bar \omega$, taking $\epsilon=10^{-3}$ and $D=4$. The right panel shows $\frac{P^{(1)}}{P^{(0)}}$ as a function of $\bar\omega$, for three different values of the dimensionless mass: $\bar m=2$ (yellow, dotted line), $\bar m=1$ (blue, dashed line) and $\bar m=0.5$ (red, continuous line). }
 \label{fig:pair_creation}
 \end{center}
\end{figure}

We have also depicted the double logarithmic plot of the ratio of pair creation assisted and non-assisted probabilities as a function of $\bar\omega$ in the right panel of Fig.~\ref{fig:pair_creation}, for a few values of the rescaled mass. For large values of $\bar\omega$, the behaviour is linear, as dictated by the prefactor in Eq.~\eqref{eq:im_ea_assisted}. Going to smaller values, there is a transition to a $\bar\omega^2$ behaviour. For small $\bar\omega$ the hierarchy of the curves with respect to the mass is kept, thanks to the following effect: in this region the softening ${\mathrm{e}}^{-\bar m ^2 \bar T^*}$ plays no role, the Lerch function becomes essentially a polylogarithm and the relevant ratio is governed by the dimensionality of the two polylogarithms.

\section{Conclusions}\label{sec:conclusion}
In this article we have shown how the Worldline Formalism can be used to perform nonperturbative computations for a scalar quantum field coupled to a Yukawa background.

In particular, using Dirichlet boundary conditions in the Worldline, we have seen that a local, resummed expression for the heat kernel is obtainable, cf. Sec.~\ref{sec:resummation}. In this way one can straightforwardly compute the first generalized heat kernel coefficients, which agree with those recently found in the literature.

Instead, in Sec.~\ref{sec:assisted_yukawa} we have developed what would be combined Parker--Toms and Barbinsky--Vilkovisky resummations, the former being used for the strong field background and the latter for the assisting field. Since the ultimate goal was to obtain the pair production probability, the so-called string inspired boundary conditions proved to be more efficient. 

From the master formula \eqref{path4}, an expansion of the integrated heat kernel to include higher derivative terms for arbitrary fields can be obtained by direct replacement. 
This enables one to compute the imaginary part of the Minkowskian effective action and, ergo, the instabilities of the vacuum in terms of the probability of pair creation. The advantage of our method over other available techniques (such as the Worldline instantons~\cite{Dunne:2006st}) is that our method provides a formula that can be computed for arbitrary, inhomogeneous fields to any desired level of precision.

Using these results, we have seen that the assistance occurs already at the first order in the assisting field. To illustrate this effect, we have worked out in detail the calculations for a Gaussian assisting field; compared to the non-assisted effect, the transseries structure of the pair creation probability is translated, leaving room for a large exponential enhancement. On the contrary, in (S)QED one has to resort at least to second order in order to observe the assisted effect~\cite{Torgrimsson:2018xdf}, because gauge invariance precludes a term like $\mathrm{e}^{-TV}$ in Eq.~\eqref{eq:HK_gaussian_assisted}.

Our work can be generalized in several directions. First, a possible new effect of pair creation in the presence of gravitational fields is under discussion~
\cite{Wondrak:2023zdi,Ferreiro:2023jfs,Hertzberg:2023xve,Akhmedov:2024axn, Boasso:2024ryt}. Using the curved spacetime version of the Wordline Formalism, such a situation is tractable.
In this scenario, we would also expect the possibility to have a first order assisted pair creation, if a mixed of strong/fastly-varying potentials is used. In this sense, our computations in this article are of great help, since at least in a perturbative approach, the computation in a curved spacetime is totally analogous to a scalar computation~\cite{vandeVen:1997pf}.
Second, the interaction of fields with higher spins is also accessible in the worldline; the corresponding resummed expressions could be already inferred from our present results.
Finally, understanding whether further resummations of derivatives are possible would be a further step in the analysis of nonperturbative phenomena, in particular related to the transseries structure of the corresponding physical effects.

\section*{Acknowledgments}
The authors are grateful to F.~Bastianelli, G.~Degli Esposti, S.~Pla, J.~Quevillon, D.~Saviot and C.~Schubert for enlightening discussions. F.F. would like to thank the Helmholtz-Zentrum Dresden-Rossendorf (Dresden, Germany) for hospitality during the preparation of most of this work.
SAF thanks the members of LAPTh, Annecy, especially J.~Quevillon, for their warm hospitality. SAF and FDM acknowledge the support from Consejo Nacional de Investigaciones Científicas y Técnicas (CONICET) through Project PIP 11220200101426CO.
SAF acknowledge the support of UNLP through Project 11/X748 and Next Generation EU through
the project ``Geometrical and Topological effects on Quantum Matter (GeTOnQuaM)'', the research activities being carried out in the framework of the INFN Research Project QGSKY. 		The authors would like to acknowledge the contribution of the COST Action CA23130. The Authors extend their appreciation to the Italian National Group of Mathematical Physics (GNFM, INdAM) for its support.

\appendix

\section{Boundary conditions and Green functions in the worldline} \label{appA}
In this appendix we will discuss some of the technicalities mentioned throughout the work regarding boundary conditions; first we will consider the case of the free operator, i.e. the pure kinetic one, and will compute the corresponding Green functions. Afterwards, we will extend the results to the Yukawa case.

\subsection{Kinetic operator and zero mode}
The kinetic term of the worldline theory that we study corresponds to the free worldline action
\begin{equation}
S_{\mathrm{kin}}=\frac12 \int_t \, \delta_{\mu\nu} \, \dot{x}^\mu(t)\dot{x}^\nu(t)=-\frac12 \int_{t_1t_2} \, x^\mu(t_1)\left[ 
\delta_{\mu\nu} \partial^2_{t_1} \delta(t_1-t_2) \right] x^\nu(t_2)\ ,
\end{equation}
where the boundary terms arising from the integration by parts vanish both for Dirichlet and periodic BC.
In particular, we are interested in the case where the trajectories are periodic, since we will see that the vanishing DBC can be reobtained as a special case.

The associated differential operator,
\begin{equation}
K_{\mu\nu}(t,t')=-\frac12 \delta_{\mu\nu} \partial^2_{t} \delta_\mathrm{P}(t-t'),
\end{equation}
is defined on the interval $t,t'\in[0,T]$ and has a zero mode on the circle: periodic boundary conditions allow for the presence of constant paths, preventing thus $K_{\mu\nu}$ to be invertible. Note that $\delta_\mathrm{P}(t-t')$ is the Dirac delta function on the space of periodic functions on $[0,T]$ \cite{Bastianelli:2002qw}; the subscript has been dropped throughout the work for notational convenience. \\

One way to solve the invertibility issue consists of factoring out the zero mode, i.e. to decompose the generic periodic path into the zero mode $\bar{x}^\mu$ and a quantum fluctuation $y^\mu(t)$, so that
\begin{equation} \label{fact}
x^\mu(t)=\bar{x}^\mu+y^\mu(t)\ .
\end{equation}
This amounts to introducing a redundant (pure gauge) field $\bar{x}^\mu$ in the worldline theory; indeed, there is now a shift symmetry, which is actually a gauge symmetry:
\begin{align}
\begin{split} \label{gauge}
\delta \bar{x}^\mu &=\xi^\mu\ , \\
\delta y^\mu(t)&=-\xi^\mu\ .
\end{split}
\end{align}
Differently to what happens in the background field method~\cite{Abbott:1980hw}, here both fields are considered dynamical.\footnote{That is, they are integrated over in the path integral and the overcounting of physically distinct configurations must be avoided by a gauge-fixing procedure.} This explains why the shift symmetry \eqref{gauge} is promoted to a \emph{gauge symmetry}, and therefore must be gauge-fixed. The gauge-fixing procedure can be performed by means of BRST methods, as analyzed in Refs.~\cite{Bastianelli:2003bg, Bastianelli:2009eh, Corradini:2018lov, Fecit:2023kah} within the Worldline Formalism. The outcome of the discussion is that one can introduce a gauge condition of the form
\begin{equation} \label{gf}
\int_t \, \rho(t)\, y(t)=0\ ,
\end{equation}
where $\rho(t)$ is usually called the \emph{background charge} (in analogy with electrostatics) and is normalized to
\begin{equation}
\int_t \, \rho(t)=1\ .
\end{equation}
The gauge-fixing \eqref{gf} allows to invert the free kinetic operator of the fluctuations $y^\mu$. It is clear that different choices of background charge will result, in general, in different boundary conditions for the quantum fluctuations and propagators. Nevertheless, the $\rho$-independence of the path integral is still guaranteed by the BRST symmetry.\footnote{To be more precise, the worldline partition function can be expressed as an integral over the zero mode
\begin{equation*}
 Z=\int {\rm d}^D\bar{x} \, z^{(\rho)}(\bar{x},\rho)\ ,
\end{equation*} 
where the partition function density $z^{(\rho)}$ may depend on the gauge-fixing choice $\rho(t)$, but this can only happen through total derivatives, which must then integrate to zero \cite{Bastianelli:2003bg}.}

More precisely, the Green function $\mathfrak{G}(t,t')$ of the operator $-\tfrac12\partial^2_t$ acting on fields constrained by the equation \eqref{gf} depends on $\rho(t)$ and satisfies
\begin{equation} \label{MS}
{-\frac12\partial_t^2 \, \mathfrak{G}(t,t')=\delta(t-t')-\rho(t)}\ .
\end{equation}
Let us discuss the most prominent choices for the background charge for the {free} theory.

\begin{itemize}
 \item \textbf{DBC} - The choice 
 \begin{equation}
 \rho(t)=\delta(t)
 \end{equation}
 is tantamount to factorizing out the zero mode as in \eqref{fact} where $\bar{x}^\mu$ is identified with a base-point in the target space along the loop. Indeed, such a background charge produces, from \eqref{gf}, vanishing Dirichlet boundary conditions for the fluctuations: 
 \begin{equation}
 y(0)=y(T)=0\ .
 \end{equation}
 In this scenario, i.e. when working in the space of functions with vanishing DBC on the interval $[0,T]$, the Dirac delta function is actually $\delta(t-t')$ and vanishes at the boundaries $t,t'=\lbrace 0, T\rbrace $ \cite{Bastianelli:2002qw}. The delta function will still be denoted as usual, since no confusion should arise. The free propagator is given by
 \begin{align}
 \begin{split} \label{freeop}
 G^{(f)}(t,t')&=-|t-t'|+\left(t+t'\right)-\frac{2t t'}{T}\ ,
 \end{split}
 \end{align}
 which can be obtained by solving
 \begin{equation}
 -\frac12\partial_\tau^2G^{(f)}(t,t')=\delta(t-t')\ ,
 \end{equation}
 with boundary conditions
 \begin{equation}
 G^{(f)}(0,t')= G^{(f)}(T,t')=0\ .
 \end{equation}
 
 \item \textbf{SI} - The choice 
 \begin{equation}
 \rho(\tau)=\frac{1}{T}
 \end{equation}
 gives rise to the so-called ``string-inspired" boundary conditions, as the factorization \eqref{fact} of the zero mode is akin to the customary practice in string theory, namely to identify $\bar{x}^\mu$ with the ``center of mass" of the worldline \eqref{CdM} and separate it from the quantum fluctuations. The resulting boundary conditions for the latter are
 \begin{equation}
 \int_t \, y(t)=0\ ,
 \end{equation}
 while the free propagator is given by
 \begin{equation} \label{freeop2}
 \mathcal{G}^{(f)}(t,t')=-|t-t'|+\frac{(t-t')^2}{T}+\frac{T}{6}\ .
 \end{equation}
 It satisfies the equation 
 \begin{equation}
 -\frac12\partial_t^2\mathcal{G}^{(f)}(t,t')=\delta(t-t')-\frac{1}{T}\ ,
 \end{equation}
 with the conditions
 \begin{align}
 \int_t \, \mathcal{G}^{(f)}(t,t')=0\ , \quad \mathcal{G}^{(f)}(0,t')= \mathcal{G}^{(f)}(T,t')\ .
 \end{align}
 Note that $\mathcal{G}^{(f)}(t,t')$ actually depends only on $t-t'$, while $G^{(f)}(t,t')$ is a true function of two variables.
\end{itemize}

\subsection{Green function with a Yukawa coupling}
Let us now consider quadratic interactions, namely the linear, one-dimensional differential operator
\begin{equation} \label{Delta'}
\Delta_{\mu\nu} (\tau,\tau')=-\frac12 \delta_{\mu\nu}\partial^2_{t_1}\delta(t_1-t_2)+2 \, \Omega^2_{\mu\nu}(y)\delta(t_1-t_2)\ .
\end{equation}
The corresponding Green functions are described in the following (we omit Lorentz indices throughout the calculations).

\begin{itemize}
 \item \textbf{DBC} - We start with the computation of Sec.~\ref{sec:resummation}, namely with the case of vanishing DBC for the fluctuations $\hat{s}$, cf. Eq.~\eqref{DBC}. The Green function $G^{\mu\nu}(t,t'):=\Delta^{-1}_{\mu\nu}$ is defined by the following differential equation,
\begin{equation} \label{greq}
\Delta(t,t') \, G(t',t'')=\delta(t-t'')\ ,
\end{equation}
together with the boundary conditions
\begin{equation}
G (0,t')=G (T,t')=0\ .
\end{equation}
We can build an Ansatz starting from the solution of the homogeneous equation associated with \eqref{greq}, namely
\begin{align}
G (t,t')=\begin{cases}
 A(t') \, \mathrm{e}^{2\Omega t}+B(t') \, \mathrm{e}^{-2\Omega t} \quad \text{if} \quad t < t' \\
C(t') \, \mathrm{e}^{2\Omega t}+D(t') \, \mathrm{e}^{-2\Omega t} \quad \text{if} \quad t' \leq t
\end{cases}\ .
\end{align}
Imposing the usual conditions on the (dis)continuity of the (first derivative of) the Green function at $t=t'$, we can relate the four unknown functions $A,B,C,D$ to get
\begin{equation}
G(t,t')=\frac{ \sinh(2\Omega t) \sinh\big(2\Omega (T-t')\big) - \Theta(t-t') \sinh(2\Omega T) \sinh\big(2\Omega (t-t')\big) }{\Omega \sinh(2\Omega T)}\ ,
\end{equation}
where one should not forget that the Green function carries Lorentz indices $G^{\mu\nu}(t,t')$. 
Note that, in the case of vanishing strong field ($\Omega \to 0$), the Green function reduces to the free propagator \eqref{freeop}.

 \item \textbf{SI} - Let us move on to the computation of Sec.~\ref{sec:assisted_yukawa}, where string-inspired boundary conditions are enforced for the fluctuations $\hat{s}$. The Green function $\mathcal{G}^{\mu\nu}(t,t')$ associated with the linear, one-dimensional differential operator $\Delta^{\mu\nu}_{\rm SI}(t)$ is defined by the differential equation
\begin{equation}
\Delta_{\rm SI}(t,t') \, \mathcal{G}(t',t'')=\delta(t'-t'')-\frac{1}{T}\ ,
\end{equation}
and must satisfy the following conditions:
\begin{align} \label{A.33}
 \int_t \, \mathcal{G}(t,t')=0\ , \quad \mathcal{G}(0,t')= \mathcal{G}(T,t')\ .
 \end{align}
The final result for the Green function reads
\begin{align}
\begin{split}
\mathcal{G} (t,t')=\frac{1}{2\Omega} \bigg[&\frac{1}{\sinh({\Omega T})}\bigg(- \frac{\sinh(\Omega T)}{\Omega T}+\cosh\Big( \Omega(T+2t-2t' ) \Big) \bigg) 
\\
&-\Theta(t-t') \,2\,\sinh \Big({2\Omega (t-t')} \Big) \bigg]\ .
\end{split}
\end{align}
As already mentioned, the SI propagator depends only on the difference of its two arguments: for instance, note that the coincidence limit reads
\begin{equation}
\mathcal{G} (t,t)=\frac{\Omega T \coth{(\Omega T)}-1}{2\Omega^2 T}\ .
\end{equation}
In the case of vanishing assisted field ($\Omega \to 0$) the Green function reduces to the free propagator \eqref{freeop2}, as expected. 
\end{itemize}

\section{Computation of the functional determinants and generalized Gel'fand--Yaglom theorem} \label{appB}
The Gel’fand–Yaglom (GY) theorem \cite{Gelfand:1959nq} and its extensions \cite{Kirsten:2003py, Kirsten:2004qv, Kirsten:2005di} have been extensively applied in the context of the worldline formalism \cite{Dunne:2006st, DegliEsposti:2022yqw, DegliEsposti:2024upq}. Let us recall the main statement of the theorem: given a one-dimensional operator defined on an interval $z \in [0,T]$ with vanishing Dirichlet boundary conditions
\begin{equation} \label{Laplace}
\left[ -\frac{{\rm d}^2}{{\rm d} z^2}+V(z) \right] \psi(z)=\lambda \, \psi(z)\ , \quad \text{with} \quad \psi(0)=\psi(T)=0\ ,
\end{equation}
there is no need to explicitly know its eigenvalues (not even one of them) if one desires to compute its functional determinant~\cite{Dunne:2007rt}. The only required information is the boundary value of the unique solution to the initial value problem
\begin{equation}
\left[ -\frac{{\rm d}^2}{{\rm d}z^2}+V(z) \right] \Phi(z)=0\ , \quad \text{with} \quad \Phi(0)=0\ , \quad \dot{\Phi}(0)=1\ ,
\end{equation}
which satisfies
\begin{equation} \label{GY1D}
\mathrm{Det} \left[ -\frac{{\rm d}^2}{{\rm d}z^2}+V(z) \right] \propto \Phi(T)\ .
\end{equation}

The result~\eqref{GY1D} can be extended to more general boundary conditions and for higher-dimensional differential operators. Indeed, consider a family of Sturm--Liouville type operators of the form
\begin{equation} \label{L1}
L_i=-\frac{{\rm d}}{{\rm d}z}\left(P_i(z)\frac{{\rm d}}{{\rm d}z}\right) \, \mathbb{1}_{r \times r} +V_i(z)
\end{equation}
where $\mathbb{1}_{r \times r}$ is the $r \times r$ identity matrix and $V_i$ can also be matrix-valued. Then, the ratio of the functional determinants of the operator $L_1$ relative to that of another operator $L_2$ of the same type can be expressed in terms of four $2r \times 2r$ matrices $M$, $N$ and $Y_i$, which are built as follows~\cite{Kirsten:2004qv}.
The matrices $Y_i$ contain the $2r$ solutions $\mathbf{u}^{(j)}(z)$ of the eigenvalue problem\footnote{We omit a subscript $i$ in the eigenfunctions for readability reasons.} $L_i \mathbf{u}^{(j)}=\lambda \mathbf{u}^{(j)}$ with $\lambda\to 0$ and their weighted first derivatives 
\begin{align}\label{eq:weighted_derivatives}
 \mathbf{v}^{(j)}(z):=P_i(z)\frac{{\rm d}}{{\rm d}z}\mathbf{u}^{(j)}(z)\ ;
\end{align}
using these vectors as (semi)-column entries, we have 
\begin{equation}
Y_i:=\left(
\begin{array}{ccc} \mathbf{u}^{(1)} & \dots & \mathbf{u}^{(2r)} \\ \mathbf{v}^{(1)} & \dots & \mathbf{v}^{(2r)} \\
\end{array}\right)\ .
\end{equation}
For simplicity, one generally chooses the initial conditions as\footnote{For instance, in the case $r=1$ we have two one-dimensional solutions $u^{1,2}(z)$ with first derivatives $v^{1,2}(z)$, whose initial conditions read $u^{(1)}(0)=1 \, , v^{(1)}(0)=0, u^{(2)}(0)=0$ and $v^{(2)}(0)=1$.}
\begin{equation} \label{BCY}
Y(0)=\mathbb{1}_{2r \times 2r}\ . 
\end{equation}
On the other hand, the matrices $M$ and $N$ are fixed by the choice of the boundary conditions; in fact, the latter can be recast in full generality as
\begin{equation}
M \, Y(0)+ N \, Y(T)=0\ .
\end{equation}
As an example, in the case of vanishing DBC these matrices reduce to
\begin{equation}
M_{\rm D}=\left(
\begin{array}{cc} \mathbb{1}_{r \times r} & \mathbb{0}_{r \times r} \\ \mathbb{0}_{r \times r} & \mathbb{0}_{r \times r} \\
\end{array}\right)\ , \quad 
N_{\rm D}=\left(
\begin{array}{cc} \mathbb{0}_{r \times r} & \mathbb{0}_{r \times r} \\ \mathbb{1}_{r \times r} & \mathbb{0}_{r \times r} \\
\end{array}\right)\ ,
\end{equation}
while in the case of periodic boundary conditions they read
\begin{equation} \label{MN}
M_{\rm P}=\left(
\begin{array}{cc} \mathbb{1}_{r \times r} & \mathbb{0}_{r \times r} \\ \mathbb{0}_{r \times r} & \mathbb{1}_{r \times r} \\
\end{array}\right)\ , \quad 
N_{\rm P}=\left(
\begin{array}{cc} -\mathbb{1}_{r \times r} & \mathbb{0}_{r \times r} \\ \mathbb{0}_{r \times r} & -\mathbb{1}_{r \times r} \\
\end{array}\right)\ .
\end{equation}

Employing this notation, the generalized GY theorem states that the ratio of the functional determinants satisfy
\begin{equation} \label{GY}
\frac{\mathrm{Det}(L_1)}{\mathrm{Det}(L_2)}=\frac{\mathrm{det}(M+N \, Y_1(T))}{\mathrm{det}(M+N \, Y_2(T))}\ .
\end{equation}
Let us now show how to use this formalism for the interacting differential operator $\Delta_{\mu\nu}$ defined in Eq.~\eqref{Delta}.

\subsection{Dirichlet boundary conditions} \label{appB.2}

For DBC, relevant to Sec.~\ref{sec:resummation}, the formula \eqref{GY} reduces to the evaluation of the determinant of an $r \times r$ solutions-submatrix:
\begin{equation}
\mathrm{Det}(\Delta)=\mathrm{det}(\mathbf{u}^{(r+1)} \, \dots \, \mathbf{u}^{(2r)})\ .
\end{equation}
Thus, we don't need to know the whole set of $2r$ solutions of the homogeneous equation; instead, we only need the $r=D$ solutions $\varphi_\mu^{(\rho)}$ which have initial conditions
\begin{equation}
\varphi_\mu^{(\rho)}(0)=\mathbb{0}\ , \quad \dot{\varphi}_\mu^{(\rho)}(0)=\delta^\rho_\mu\ .
\end{equation}
Building from the one-dimensional solution an Ansatz of the form
\begin{equation}
\varphi_\mu^{(\rho)}(z)=\left(A \, \mathrm{e}^{2\Omega z} \right)^\rho{}_\mu+\left(B \, \mathrm{e}^{-2\Omega z} \right)^\rho{}_\mu\ ,
\end{equation}
one obtains
\begin{equation}
\varphi_\mu^{(\rho)}(z)=\left[ \left.\frac{1}{2\Omega}\sinh(2\Omega z) \right]^\rho{}\right._\mu\ .
\end{equation}
Using as a reference operator the free kinetic operator $L_2=-\tfrac12\delta_{\mu\nu}\partial^2_z$, whose $Y$-matrix can be effortless computed, the final result for the functional determinant reads
\begin{equation}
{\overline{\mathrm{Det}}(\Delta)=\mathrm{det}\left[ \frac{\sinh(2\Omega T)}{2\Omega T} \right]}\ .
\end{equation}
As a double check of this result, we can calculate the determinant as a $\zeta$-regularized infinite product of the eigenvalues of the operator $\Delta$ divided by those of the free operator. Using a simple Fourier expansion to solve the eigenvalue equation, we get
\begin{equation}
 \overline{\mathrm{Det}}(\Delta)= \prod_{j=1}^{r} \prod\limits_{n =1}^{\infty} \left[ 1+\left(\frac{2\Omega^{(j)} T}{\pi n}\right)^2\right] =\mathrm{det}\left[ 
\frac{\sinh{(2\Omega T)}}{2\Omega T} \right]\ ,
\end{equation}
where $\Omega^{(j)}$ are the eigenvalues of the matrix $\Omega^{\mu}{}_{\nu}$ and we have used the well-known result
\begin{equation}
 \prod\limits_{n=1}^{\infty}\left( 1+\frac{x^2}{\pi^2n^2} \right)=\frac{\sinh{(x)}}{x}\ .
\end{equation}
Note in particular that the free limit is regular:
\begin{equation}
 \lim_{\Omega\to0} \, \overline{\mathrm{Det}}(\Delta)=1\ .
\end{equation}
\subsection{Periodic boundary conditions} \label{app:det_pbc}

The PBC, relevant to Sec.~\ref{sec:assisted_yukawa}, require more carefulness; let us first show the exemplificative non-matricial case, i.e. $r=1$. 

\newtheorem{theorem}{Theorem}
\newtheorem{proposition}{Proposition} 

\begin{proposition}

Assuming Eq.~\eqref{MN} for the matrices $M$ and $N$, and considering an arbitrary operator of the form given in Eq.~\eqref{L1}, the relevant determinant for the GY formula in the one-dimensional ($r=1$) case reduces to
\begin{align}
\begin{split} \label{r=1}
\mathrm{det}\Big(M+N \, Y_1(T)\Big)
&=2-\mathrm{tr}\left(Y_{r=1}(T)\right)
= 2-2\cosh(2\Omega T) \ .
\end{split}
\end{align}
\end{proposition}

\begin{proof}
In general, a direct calculation gives actually
\begin{align}
\begin{split}
\mathrm{det}\Big(M+N \, Y_1(T)\Big)&=1-\left.\left( u^{(1)}+v^{(2)} \right)\right|_{z=T}+\left.(u^{(1)}v^{(2)}-v^{(1)}u^{(2)})\right|_{z=T} \label{B20} \\
&=1- \mathrm{tr}\left(Y_{r=1}(T) \right)+\mathrm{det}\left( Y_{r=1} (T)\right)\ ,
\end{split}
\end{align}
with the matrix of solutions given by
\begin{equation}
Y_{r=1}(z)=\left(
\begin{array}{cc} u^{(1)}(z) & u^{(2)}(z) \\ v^{(1)}(z) & v^{(2)}(z) \\
\end{array}\right)\ .
\end{equation}
One can recognize in the RHS of \eqref{B20} the Wrosnkian
\begin{equation}
\mathcal{W}(u^{(1)},u^{(2)})(z)=u^{(1)}(z)v^{(2)}(z)-v^{(1)}(z)u^{(2)}(z)\ ,
\end{equation}
evaluated at $z=T$, which is associated with the two independent solutions $u^{1,2}(z)$ of the eigenvalue equation for a Sturm--Liouville operator.
It is a stablished result that this type of Wronskian is constant;
therefore, extracting from the initial conditions its value at $z=0$,
\begin{equation}
\mathcal{W}(u^{(1)},u^{(2)})(0)=1\ ,
\end{equation}
we also have that $\mathcal{W}(u^{(1)},u^{(2)})(T)=1$; hence, using the explicit solutions $u^{(1)}$ and $v^{(2)}$, Eq.~\eqref{r=1} follows.
\end{proof}

It is not hard to convince oneself that the previous one-dimensional result can be directly generalized to higher dimensions.

\begin{theorem}\label{th:GY_periodic}
Consider an operator of the form given in Eq.~\eqref{eq:Delta_SI}, defined on the domain of periodic functions. Be $u_a^{(b)}$ the $a$th component of the $b$th solution of the homogeneous equation $L_1 \mathbf{u}^{(b)}=0$ and $\mathbf{v}^{(b)}$ its corresponding weighted derivative, as in Eq.~\eqref{eq:weighted_derivatives}. 
Then, the relevant determinant for the GY theorem satisfies
\begin{align}\label{eq:theorem}
\mathrm{det}\Big(M+N \, Y_{\Delta, \mathrm{PBC}}(T)\Big)
=\det \big(2-2\cosh(2\Omega T)\big)\ .
\end{align}
\end{theorem}

\begin{proof}
Note first that, in contrast to the DBC, to compute the functional determinant using the generalized GY theorem we do need the whole set of $2r$ solutions of the homogeneous equation
\begin{equation}
\Delta_{\rm SI}^{\mu\nu} \, \Phi_\nu^{(\Pi)}=0\ , \quad \Phi^{(\Pi)}_\mu=\begin{cases}
 \phi_\mu^{(\rho)} \quad \text{if} \quad \Pi=1,\dots,r \\
 \psi_\mu^{(\rho)} \quad \text{if} \quad \Pi=r+1,\dots,2r
\end{cases}\ .
\end{equation}
 Considering Eq.~\eqref{BCY} as initial conditions, namely 
\begin{align}
\phi_\mu^{(\rho)}(0)&=\delta^\rho_\mu\ , \quad \frac{{\rm d}\phi}{{\rm d}z}_\mu^{(\rho)}(0)=\mathbb{0}\ ,
\\
\psi_\mu^{(\rho)}(0)&=\mathbb{0}\ , \quad \frac{{\rm d}\psi}{{\rm d}z}_\mu^{(\rho)}(0)=\delta^\rho_\mu\ ,
\end{align}
an explicit computation gives
\begin{equation}
\phi_\mu^{(\rho)}(z)=\dot{\psi}_\mu^{(\rho)}(z)=\Big[\cosh(2\Omega z)\Big]^\rho{}\Big._\mu\ .
\end{equation}
Note then that in arbitrary dimensions we can straightforwardly prove that the $n$th power of the matrix $Y_\Delta$ is
\begin{align}
 Y^n_{\Delta, \mathrm{PBC}} (T)= \begin{pmatrix}
 \cosh(2n\Omega T) & \sinh(2n\Omega T)
 \\
 \sinh(2n\Omega T) & \cosh(2n\Omega T)
 \end{pmatrix},
\end{align}
and consequently, from Eq.~\eqref{GY}, we have the equalities
\begin{align}
 \begin{split}
 \det(\mathbb{1}_{2r} - Y_{\Delta, \mathrm{PBC}}) 
 &= \mathrm{e}^{\tr \log [2-2 \cosh(2\Omega T)]}
 = \det \big(2-2\cosh(2\Omega T)\big) \ ,
 \end{split}
\end{align}
which indeed agree with the RHS of Eq.~\eqref{eq:theorem}. 
\end{proof}
As a corollary of Th.~\ref{th:GY_periodic}, recall that the difference between P and SI BC resides just in the omission of the constant mode in the latter; in the free case, this mode corresponds to a zero mode. The quotient of $\Delta_{\rm SI}$ with the the free operator, after extracting the zero mode of the latter (denoted with a prime in the Det), is thus given by
\begin{equation} \label{Rr}
{\overline{\mathrm{Det}}'(\Delta_{\rm SI})=\mathrm{det}\left[ \frac{\sinh^2(\Omega T)}{\Omega^2T^2} \right]}\ .
\end{equation}
This is consistent with the fact that we have omitted the constant mode, so that the free limit $\Omega\to 0$ is nonvanishing. As an alternative check of the GY result, we can directly evaluate the determinant of $\Delta_{\rm SI}$ as a product of its eigenvalues; using a Fourier expansion and dividing by the free eigenvalues (we exclude the vanishing one), we get the following $\zeta$-regularized infinite product
\begin{equation} \label{B.48}
 \overline{\mathrm{Det}}'
 (\Delta_{\rm SI})= \mathrm{det}\left[ \prod\limits_{n =1}^{\infty} \left( 1+\frac{\Omega^2T^2}{\pi^2n^2} \right)^2 \, \right]=\mathrm{det}\left[ \frac{\sinh^2(\Omega T)}{\Omega^2 T^2} \right]\ ,
\end{equation}
which confirms Eq.~\eqref{Rr}.

\section{Generalized Schwinger--DeWitt coefficients in the string inspired approach}\label{app:coeff_SI}
The first generalized Schwinger--DeWitt coefficients computed using the SI BC, defined in the expansion
\begin{align}
    \Sigma^{(1)}_0({x};T) =:\sum_{j=0}^{\infty} c^{\mathrm{SI}}_j(x) \, T^{j}\ ,
\end{align}
are given by
\begin{align}
c^{\mathrm{SI}}_0(x)&=0\ , \\
c^{\mathrm{SI}}_1(x)&=0\ , \\
c^{\mathrm{SI}}_2(x)&=0\ , \\
c^{\mathrm{SI}}_3(x)&=-\frac{1}{288}\partial_{\mu}{}^{\mu}{}_{\nu}{}^{\nu} V  \ ,
\\
c^{\mathrm{SI}}_4(x)&=-\frac{1}{10368}\partial_{\mu}{}^\mu{}_\nu{}^\nu{}_\rho{}^\rho V \ ,
\\
c^{\mathrm{SI}}_5(x)&= 
\frac{1}{45360} \partial_{\mu\nu\rho} V \partial^{\mu\nu\rho} V  
+ \frac{1}{4320} \partial_{\mu}{}^{\mu}{}_{\rho\lambda} V \partial^{\rho\lambda} V 
-\frac{1}{497664}\partial_{\mu}{}^\mu{}_\nu{}^\nu{}_\rho{}^\rho{}_\sigma{}^\sigma V  \ .
\end{align}
Recall that these coefficients are only valid  at the integrated level, i.e. they are related via integration by parts to the $c_i$ coefficients in Eq.~\eqref{eq:generalized_coeff}.

\addcontentsline{toc}{section}{References}
\bibliography{biblio.bib}

\providecommand{\href}[2]{#2}\begingroup\raggedright\begin{thebibliography}{10}

\bibitem{Euler:1935zz}
H.~Euler and B.~Kockel, \emph{{The scattering of light by light in Dirac\textquoteright{}s theory}}, \href{https://doi.org/10.1007/BF01493898}{\emph{Naturwiss.} {\bfseries 23} (1935) 246}.

\bibitem{Heisenberg:1936nmg}
W.~Heisenberg and H.~Euler, \emph{{Consequences of Dirac's theory of positrons}}, \href{https://doi.org/10.1007/BF01343663}{\emph{Z. Phys.} {\bfseries 98} (1936) 714} [\href{https://arxiv.org/abs/physics/0605038}{{\ttfamily physics/0605038}}].

\bibitem{Schwinger:1951nm}
J.S.~Schwinger, \emph{{On gauge invariance and vacuum polarization}}, \href{https://doi.org/10.1103/PhysRev.82.664}{\emph{Phys. Rev.} {\bfseries 82} (1951) 664}.

\bibitem{Ahmadiniaz:2024xob}
N.~Ahmadiniaz et~al., \emph{{Letter of Intent: Towards a Vacuum Birefringence Experiment at the Helmholtz International Beamline for Extreme Fields}},  \href{https://arxiv.org/abs/2405.18063}{{\ttfamily 2405.18063}}.

\bibitem{LUXE:2023crk}
{\scshape LUXE} collaboration, \emph{{Technical Design Report for the LUXE Experiment}},  \href{https://arxiv.org/abs/2308.00515}{{\ttfamily 2308.00515}}.

\bibitem{Kraych:2024wwd}
A.E.~Kraych et~al., \emph{{Interferometric measurement of the deflection of light by light in air}}, \href{https://doi.org/10.1103/PhysRevA.109.053510}{\emph{Phys. Rev. A} {\bfseries 109} (2024) 053510} [\href{https://arxiv.org/abs/2401.13506}{{\ttfamily 2401.13506}}].

\bibitem{Fan:2017fnd}
X.~Fan et~al., \emph{{The OVAL experiment: A new experiment to measure vacuum magnetic birefringence using high repetition pulsed magnets}}, \href{https://doi.org/10.1140/epjd/e2017-80290-7}{\emph{Eur. Phys. J. D} {\bfseries 71} (2017) 308} [\href{https://arxiv.org/abs/1705.00495}{{\ttfamily 1705.00495}}].

\bibitem{Schmitt:2022pkd}
A.~Schmitt et~al., \emph{{Mesoscopic Klein-Schwinger effect in graphene}}, \href{https://doi.org/10.1038/s41567-023-01978-9}{\emph{Nature Phys.} {\bfseries 19} (2023) 830} [\href{https://arxiv.org/abs/2207.13400}{{\ttfamily 2207.13400}}].

\bibitem{Dunne:2022esi}
G.V.~Dunne and Z.~Harris, \emph{{Resurgence of the effective action in inhomogeneous fields}}, \href{https://doi.org/10.1103/PhysRevD.107.065003}{\emph{Phys. Rev. D} {\bfseries 107} (2023) 065003} [\href{https://arxiv.org/abs/2212.04599}{{\ttfamily 2212.04599}}].

\bibitem{Karbstein:2023yee}
F.~Karbstein, \emph{{Towards the full Heisenberg-Euler effective action at large N}}, \href{https://doi.org/10.1007/JHEP02(2024)180}{\emph{JHEP} {\bfseries 02} (2024) 180} [\href{https://arxiv.org/abs/2312.09804}{{\ttfamily 2312.09804}}].

\bibitem{Karbstein:2021gdi}
F.~Karbstein, \emph{{Large N external-field quantum electrodynamics}}, \href{https://doi.org/10.1007/JHEP01(2022)057}{\emph{JHEP} {\bfseries 01} (2022) 057} [\href{https://arxiv.org/abs/2109.04823}{{\ttfamily 2109.04823}}].

\bibitem{Copinger:2024pai}
P.~Copinger, J.P.~Edwards, A.~Ilderton and K.~Rajeev, \emph{{Pair creation, backreaction, and resummation in strong fields}},  \href{https://arxiv.org/abs/2411.06203}{{\ttfamily 2411.06203}}.

\bibitem{Franchino-Vinas:2023wea}
S.A.~Franchino-Vi\~nas, C.~Garc\'\i{}a-P\'erez, F.D.~Mazzitelli, V.~Vitagliano and U.W.~Haimovichi, \emph{{Resummed heat kernel and effective action for Yukawa and QED}}, \href{https://doi.org/10.1016/j.physletb.2024.138684}{\emph{Phys. Lett. B} {\bfseries 854} (2024) 138684} [\href{https://arxiv.org/abs/2312.16303}{{\ttfamily 2312.16303}}].

\bibitem{Navarro-Salas:2020oew}
J.~Navarro-Salas and S.~Pla, \emph{{$(\mathcal{F},\mathcal{G})$-summed form of the QED effective action}}, \href{https://doi.org/10.1103/PhysRevD.103.L081702}{\emph{Phys. Rev. D} {\bfseries 103} (2021) L081702} [\href{https://arxiv.org/abs/2011.09743}{{\ttfamily 2011.09743}}].

\bibitem{Schubert:2001he}
C.~Schubert, \emph{{Perturbative quantum field theory in the string inspired formalism}}, \href{https://doi.org/10.1016/S0370-1573(01)00013-8}{\emph{Phys. Rept.} {\bfseries 355} (2001) 73} [\href{https://arxiv.org/abs/hep-th/0101036}{{\ttfamily hep-th/0101036}}].

\bibitem{Bastianelli:2024vkp}
F.~Bastianelli, O.~Corradini, J.P.~Edwards, D.G.C.~McKeon and C.~Schubert, \emph{{Unified worldline treatment of Yukawa and axial couplings}}, \href{https://doi.org/10.1007/JHEP11(2024)152}{\emph{JHEP} {\bfseries 11} (2024) 152} [\href{https://arxiv.org/abs/2406.19988}{{\ttfamily 2406.19988}}].

\bibitem{Dunne:2006st}
G.V.~Dunne, Q.-h.~Wang, H.~Gies and C.~Schubert, \emph{{Worldline instantons. II. The Fluctuation prefactor}}, \href{https://doi.org/10.1103/PhysRevD.73.065028}{\emph{Phys. Rev. D} {\bfseries 73} (2006) 065028} [\href{https://arxiv.org/abs/hep-th/0602176}{{\ttfamily hep-th/0602176}}].

\bibitem{DegliEsposti:2024upq}
G.~Degli~Esposti and G.~Torgrimsson, \emph{{Schwinger pair production in spacetime fields: Moir\'e patterns, Aharonov-Bohm phases and Sturm-Liouville eigenvalues}},  \href{https://arxiv.org/abs/2412.19709}{{\ttfamily 2412.19709}}.

\bibitem{Copinger:2024twl}
P.~Copinger, J.P.~Edwards, A.~Ilderton and K.~Rajeev, \emph{{All-multiplicity amplitudes in impulsive PP-waves from the worldline formalism}}, \href{https://doi.org/10.1007/JHEP09(2024)148}{\emph{JHEP} {\bfseries 09} (2024) 148} [\href{https://arxiv.org/abs/2405.07385}{{\ttfamily 2405.07385}}].

\bibitem{Ahumada:2023iac}
I.~Ahumada and J.P.~Edwards, \emph{{Monte Carlo generation of localized particle trajectories}}, \href{https://doi.org/10.1103/PhysRevE.108.065306}{\emph{Phys. Rev. E} {\bfseries 108} (2023) 065306} [\href{https://arxiv.org/abs/2304.10518}{{\ttfamily 2304.10518}}].

\bibitem{Franchino-Vinas:2019udt}
S.~Franchino-Vi\~nas and H.~Gies, \emph{{Propagator from Nonperturbative Worldline Dynamics}}, \href{https://doi.org/10.1103/PhysRevD.100.105020}{\emph{Phys. Rev. D} {\bfseries 100} (2019) 105020} [\href{https://arxiv.org/abs/1908.04532}{{\ttfamily 1908.04532}}].

\bibitem{Parker:1984dj}
L.~Parker and D.J.~Toms, \emph{{New Form for the Coincidence Limit of the Feynman Propagator, or Heat Kernel, in Curved Space-time}}, \href{https://doi.org/10.1103/PhysRevD.31.953}{\emph{Phys. Rev. D} {\bfseries 31} (1985) 953}.

\bibitem{Hu:1984js}
B.L.~Hu and D.J.~O'Connor, \emph{{Effective Lagrangian for $\lambda \phi^4$ Theory in Curved Space-time With Varying Background Fields: Quasilocal Approximation}}, \href{https://doi.org/10.1103/PhysRevD.30.743}{\emph{Phys. Rev. D} {\bfseries 30} (1984) 743}.

\bibitem{Jack:1985mw}
I.~Jack and L.~Parker, \emph{{Proof of Summed Form of Proper Time Expansion for Propagator in Curved Space-time}}, \href{https://doi.org/10.1103/PhysRevD.31.2439}{\emph{Phys. Rev. D} {\bfseries 31} (1985) 2439}.

\bibitem{Flachi:2015sva}
A.~Flachi, K.~Fukushima and V.~Vitagliano, \emph{{Geometrically induced magnetic catalysis and critical dimensions}}, \href{https://doi.org/10.1103/PhysRevLett.114.181601}{\emph{Phys. Rev. Lett.} {\bfseries 114} (2015) 181601} [\href{https://arxiv.org/abs/1502.06090}{{\ttfamily 1502.06090}}].

\bibitem{Schutzhold:2008pz}
R.~Schutzhold, H.~Gies and G.~Dunne, \emph{{Dynamically assisted Schwinger mechanism}}, \href{https://doi.org/10.1103/PhysRevLett.101.130404}{\emph{Phys. Rev. Lett.} {\bfseries 101} (2008) 130404} [\href{https://arxiv.org/abs/0807.0754}{{\ttfamily 0807.0754}}].

\bibitem{Barvinsky:1987uw}
A.~Barvinsky and G.~Vilkovisky, \emph{{Beyond the Schwinger-Dewitt Technique: Converting Loops Into Trees and In-In Currents}}, \href{https://doi.org/10.1016/0550-3213(87)90681-X}{\emph{Nucl. Phys. B} {\bfseries 282} (1987) 163}.

\bibitem{Barvinsky:1990up}
A.O.~Barvinsky and G.A.~Vilkovisky, \emph{{Covariant perturbation theory. 2: Second order in the curvature. General algorithms}}, \href{https://doi.org/10.1016/0550-3213(90)90047-H}{\emph{Nucl. Phys. B} {\bfseries 333} (1990) 471}.

\bibitem{Silva:2023lts}
W.C.e.~Silva and I.L.~Shapiro, \emph{{Effective approach to the Antoniadis-Mottola model: quantum decoupling of the higher derivative terms}}, \href{https://doi.org/10.1007/JHEP07(2023)097}{\emph{JHEP} {\bfseries 07} (2023) 097} [\href{https://arxiv.org/abs/2301.13291}{{\ttfamily 2301.13291}}].

\bibitem{DeWitt:2003pm}
B.S.~DeWitt, \emph{{The global approach to quantum field theory. Vol. 1, 2}}, Clarendon Press (7, 2003).

\bibitem{Jeffreys}
H.~Jeffreys and B.~Jeffreys, \emph{{Methods of Mathematical Physics}}, Cambridge University Press, Cambridge, England, third~ed. (1988).

\bibitem{Vinas:2014exa}
S.F.~Vi\~nas and P.~Pisani, \emph{{Worldline approach to the Grosse-Wulkenhaar model}}, \href{https://doi.org/10.1007/JHEP11(2014)087}{\emph{JHEP} {\bfseries 11} (2014) 087} [\href{https://arxiv.org/abs/1406.7336}{{\ttfamily 1406.7336}}].

\bibitem{Franchino-Vinas:2021bcl}
S.A.~Franchino-Vinas and S.~Mignemi, \emph{{The Snyder-de Sitter scalar $\varphi^4_\star$ quantum field theory in $D=2$}}, \href{https://doi.org/10.1016/j.nuclphysb.2022.115871}{\emph{Nucl. Phys. B} {\bfseries 981} (2022) 115871} [\href{https://arxiv.org/abs/2104.00043}{{\ttfamily 2104.00043}}].

\bibitem{Gelfand:1959nq}
I.M.~Gelfand and A.M.~Yaglom, \emph{{Integration in functional spaces and it applications in quantum physics}}, \href{https://doi.org/10.1063/1.1703636}{\emph{J. Math. Phys.} {\bfseries 1} (1960) 48}.

\bibitem{Kirsten:2003py}
K.~Kirsten and A.J.~McKane, \emph{{Functional determinants by contour integration methods}}, \href{https://doi.org/10.1016/S0003-4916(03)00149-0}{\emph{Annals Phys.} {\bfseries 308} (2003) 502} [\href{https://arxiv.org/abs/math-ph/0305010}{{\ttfamily math-ph/0305010}}].

\bibitem{Kirsten:2004qv}
K.~Kirsten and A.J.~McKane, \emph{{Functional determinants for general Sturm-Liouville problems}}, \href{https://doi.org/10.1088/0305-4470/37/16/014}{\emph{J. Phys. A} {\bfseries 37} (2004) 4649} [\href{https://arxiv.org/abs/math-ph/0403050}{{\ttfamily math-ph/0403050}}].

\bibitem{Franchino-Vinas:2024wof}
S.A.~Franchino-Vi\~nas, \emph{{Comment on \textquoteleft{}Index-free heat kernel coefficients\textquoteright{}}}, \href{https://doi.org/10.1088/1361-6382/ad4949}{\emph{Class. Quant. Grav.} {\bfseries 41} (2024) 128001} [\href{https://arxiv.org/abs/2401.01296}{{\ttfamily 2401.01296}}].

\bibitem{Torgrimsson:2017pzs}
G.~Torgrimsson, C.~Schneider, J.~Oertel and R.~Sch\"utzhold, \emph{{Dynamically assisted Sauter-Schwinger effect \textemdash{} non-perturbative versus perturbative aspects}}, \href{https://doi.org/10.1007/JHEP06(2017)043}{\emph{JHEP} {\bfseries 06} (2017) 043} [\href{https://arxiv.org/abs/1703.09203}{{\ttfamily 1703.09203}}].

\bibitem{Torgrimsson:2018xdf}
G.~Torgrimsson, \emph{{Perturbative methods for assisted nonperturbative pair production}}, \href{https://doi.org/10.1103/PhysRevD.99.096002}{\emph{Phys. Rev. D} {\bfseries 99} (2019) 096002} [\href{https://arxiv.org/abs/1812.04607}{{\ttfamily 1812.04607}}].

\bibitem{Vassilevich:2003xt}
D.V.~Vassilevich, \emph{{Heat kernel expansion: User's manual}}, \href{https://doi.org/10.1016/j.physrep.2003.09.002}{\emph{Phys. Rept.} {\bfseries 388} (2003) 279} [\href{https://arxiv.org/abs/hep-th/0306138}{{\ttfamily hep-th/0306138}}].

\bibitem{Fedotov:2022ely}
A.~Fedotov, A.~Ilderton, F.~Karbstein, B.~King, D.~Seipt, H.~Taya et~al., \emph{{Advances in QED with intense background fields}}, \href{https://doi.org/10.1016/j.physrep.2023.01.003}{\emph{Phys. Rept.} {\bfseries 1010} (2023) 1} [\href{https://arxiv.org/abs/2203.00019}{{\ttfamily 2203.00019}}].

\bibitem{Keldysh:1965ojf}
L.V.~Keldysh, \emph{{Ionization in the Field of a Strong Electromagnetic Wave}}, {\emph{J. Exp. Theor. Phys.} {\bfseries 20} (1965) 1307}.

\bibitem{Wondrak:2023zdi}
M.F.~Wondrak, W.D.~van Suijlekom and H.~Falcke, \emph{{Gravitational Pair Production and Black Hole Evaporation}}, \href{https://doi.org/10.1103/PhysRevLett.130.221502}{\emph{Phys. Rev. Lett.} {\bfseries 130} (2023) 221502} [\href{https://arxiv.org/abs/2305.18521}{{\ttfamily 2305.18521}}].

\bibitem{Ferreiro:2023jfs}
A.~Ferreiro, J.~Navarro-Salas and S.~Pla, \emph{{Comment on ``Gravitational Pair Production and Black Hole Evaporation''}},  \href{https://arxiv.org/abs/2306.07628}{{\ttfamily 2306.07628}}.

\bibitem{Hertzberg:2023xve}
M.P.~Hertzberg and A.~Loeb, \emph{{Inconsistency with De Sitter Spacetime of ''Gravitational Pair Production and Black Hole Evaporation''}},  \href{https://arxiv.org/abs/2307.05243}{{\ttfamily 2307.05243}}.

\bibitem{Akhmedov:2024axn}
E.T.~Akhmedov, D.V.~Diakonov and C.~Schubert, \emph{{Complex effective actions and gravitational pair creation}},  \href{https://arxiv.org/abs/2407.06601}{{\ttfamily 2407.06601}}.

\bibitem{Boasso:2024ryt}
A.~Boasso, S.~Franchino-Vi\~nas and F.D.~Mazzitelli, \emph{{Nonlocal effective action and particle creation in $D$ dimensions}},  \href{https://arxiv.org/abs/2412.03340}{{\ttfamily 2412.03340}}.

\bibitem{vandeVen:1997pf}
A.E.M.~van~de Ven, \emph{{Index free heat kernel coefficients}}, \href{https://doi.org/10.1088/0264-9381/15/8/014}{\emph{Class. Quant. Grav.} {\bfseries 15} (1998) 2311} [\href{https://arxiv.org/abs/hep-th/9708152}{{\ttfamily hep-th/9708152}}].

\bibitem{Bastianelli:2002qw}
F.~Bastianelli, O.~Corradini and A.~Zirotti, \emph{{dimensional regularization for N=1 supersymmetric sigma models and the worldline formalism}}, \href{https://doi.org/10.1103/PhysRevD.67.104009}{\emph{Phys. Rev. D} {\bfseries 67} (2003) 104009} [\href{https://arxiv.org/abs/hep-th/0211134}{{\ttfamily hep-th/0211134}}].

\bibitem{Abbott:1980hw}
L.F.~Abbott, \emph{{The Background Field Method Beyond One Loop}}, \href{https://doi.org/10.1016/0550-3213(81)90371-0}{\emph{Nucl. Phys. B} {\bfseries 185} (1981) 189}.

\bibitem{Bastianelli:2003bg}
F.~Bastianelli, O.~Corradini and A.~Zirotti, \emph{{BRST treatment of zero modes for the worldline formalism in curved space}}, \href{https://doi.org/10.1088/1126-6708/2004/01/023}{\emph{JHEP} {\bfseries 01} (2004) 023} [\href{https://arxiv.org/abs/hep-th/0312064}{{\ttfamily hep-th/0312064}}].

\bibitem{Bastianelli:2009eh}
F.~Bastianelli, O.~Corradini and A.~Waldron, \emph{{Detours and Paths: BRST Complexes and Worldline Formalism}}, \href{https://doi.org/10.1088/1126-6708/2009/05/017}{\emph{JHEP} {\bfseries 05} (2009) 017} [\href{https://arxiv.org/abs/0902.0530}{{\ttfamily 0902.0530}}].

\bibitem{Corradini:2018lov}
O.~Corradini and M.~Muratori, \emph{{String-inspired Methods and the Worldline Formalism in Curved Space}}, \href{https://doi.org/10.1140/epjp/i2018-12293-5}{\emph{Eur. Phys. J. Plus} {\bfseries 133} (2018) 457} [\href{https://arxiv.org/abs/1808.05401}{{\ttfamily 1808.05401}}].

\bibitem{Fecit:2023kah}
F.~Fecit, \emph{{Massive gravity from a first-quantized perspective}}, \href{https://doi.org/10.1140/epjc/s10052-024-12799-2}{\emph{Eur. Phys. J. C} {\bfseries 84} (2024) 445} [\href{https://arxiv.org/abs/2312.15428}{{\ttfamily 2312.15428}}].

\bibitem{Kirsten:2005di}
K.~Kirsten and A.J.~McKane, \emph{{Functional determinants in the presence of zero modes}},  in \emph{{6th Workshop on Quantum Field Theory under the Influence of External Conditions (QFEXT03)}}, pp.~146--151, 7, 2005 [\href{https://arxiv.org/abs/hep-th/0507005}{{\ttfamily hep-th/0507005}}].

\bibitem{DegliEsposti:2022yqw}
G.~Degli~Esposti and G.~Torgrimsson, \emph{{Worldline instantons for the momentum spectrum of Schwinger pair production in spacetime dependent fields}}, \href{https://doi.org/10.1103/PhysRevD.107.056019}{\emph{Phys. Rev. D} {\bfseries 107} (2023) 056019} [\href{https://arxiv.org/abs/2212.11578}{{\ttfamily 2212.11578}}].

\bibitem{Dunne:2007rt}
G.V.~Dunne, \emph{{Functional determinants in quantum field theory}}, \href{https://doi.org/10.1088/1751-8113/41/30/304006}{\emph{J. Phys. A} {\bfseries 41} (2008) 304006} [\href{https://arxiv.org/abs/0711.1178}{{\ttfamily 0711.1178}}].

\end{thebibliography}\endgroup


\end{document}